\documentclass{llncs}
\def\Draft{0}

\usepackage[utf8x]{inputenc}
\usepackage{amsmath}
\usepackage{amsfonts}
\usepackage{algorithm}
\usepackage{algpseudocode}
\usepackage{graphicx}
\usepackage{siunitx}
\usepackage{placeins}
\usepackage{lmodern}
\usepackage{array}
\usepackage{booktabs}
\usepackage{amssymb}
\usepackage{multirow}
\usepackage[toc,page,titletoc]{appendix}
\usepackage{breakcites}
\usepackage[hidelinks,colorlinks=true]{hyperref}
\usepackage{tabularx,booktabs,ragged2e} 
\usepackage{tikz}
\usetikzlibrary{matrix,decorations.pathreplacing, calc, positioning}
\usepackage{xifthen}
\usepackage{array}

\newcolumntype{L}[1]{>{\raggedright\let\newline\\\arraybackslash\hspace{0pt}}m{#1}}
\newcolumntype{C}[1]{>{\centering\let\newline\\\arraybackslash\hspace{0pt}}m{#1}}
\newcolumntype{R}[1]{>{\raggedleft\let\newline\\\arraybackslash\hspace{0pt}}m{#1}}

\newcommand{\transp}[1]{{#1}^{ {\intercal} }}
\newcommand{\eqdef}{\triangleq}

\newtheorem{heuristic}{Heuristic}
\newtheorem{notation}{Notation}

\usepackage{textcomp}

\newcommand{\fq}{\mathbb{F}_q}

\newcommand\F{\mathbb{F}}
\newcommand\E{\mathbb{E}}
\newcommand\Z{\mathbb{Z}}
\newcommand\R{\mathbb{R}}

\newcommand\LL{\mathcal{L}}
\newcommand\CC{\mathcal{C}}

\newcommand{\ResSD}{\textsf{ResSD}}
\newcommand{\RegSD}{\textsf{RegSD}}
\newcommand{\ListSVP}{\textsf{List-SVP}}
\newcommand{\ListCVP}{\textsf{List-CVP}}
\newcommand{\SVP}{\textsf{SVP}}
\newcommand{\CVP}{\textsf{CVP}}
\newcommand{\Ball}[1]{\mathcal{B}_{#1}}

\newcommand{\av}{{\mathbf{a}}}
\newcommand{\bv}{{\mathbf{b}}}
\newcommand{\bvs}{\mathbf{b}^*}
\newcommand{\cv}{{\mathbf{c}}}
\newcommand{\ev}{{\mathbf{e}}}

\newcommand{\kv}{{\mathbf{k}}}

\newcommand{\sv}{{\mathbf{s}}}
\newcommand{\uv}{{\mathbf{u}}}
\newcommand{\vv}{{\mathbf{v}}}
\newcommand{\xv}{{\mathbf{x}}}
\newcommand{\yv}{{\mathbf{y}}}

\newcommand{\Am}{{\mathbf{A}}}
\newcommand{\Bm}{{\mathbf{B}}}
\newcommand{\Cm}{{\mathbf{C}}}
\newcommand{\Gm}{{\mathbf{G}}}
\newcommand{\Hm}{{\mathbf{H}}}
\newcommand{\Imat}{{\mathbf{I}}}
\newcommand{\Mm}{{\mathbf{M}}}
\newcommand{\Pm}{{\mathbf{P}}}
\newcommand{\Qm}{{\mathbf{Q}}}
\newcommand{\Rm}{{\mathbf{R}}}
\newcommand{\Nm}{{\mathbf{N}}}

\newcommand{\Um}{{\mathbf{U}}}

\newcommand{\Omat}{{\mathbf{0}}}

\DeclareMathOperator\Vol{Vol}
\DeclareMathOperator\Span{span}
\DeclareMathOperator*\Prob{\mathbb{P}}

\newcommand{\pround}[1]{\left({#1}\right)}

\newcommand{\pset}[1]{\left\{{#1}\right\}}

\newcommand{\hnorm}[1]{\left\vert{#1}\right\vert_{\textsf{H}}}
\newcommand{\enorm}[1]{\left\vert{#1}\right\vert_{\textsf{E}}}
\newcommand{\card}[1]{\textsf{Card}\left(#1\right)}

\author{\'Etienne Burle \and Aleksei Udovenko}

\institute{
SnT, University of Luxembourg
\\\email{firstname.lastname@uni.lu}
}

\if\Draft1
\newcommand{\TODO}[1]{\textcolor{purple}{TODO: #1}}
\newcommand{\AU}[1]{\textcolor{blue!70!black}{\it AU: #1}}
\newcommand{\EB}[1]{\textcolor{green!50!black}{\it EB: #1}}
\else
\newcommand{\TODO}[1]{}
\newcommand{\AU}[1]{}
\newcommand{\EB}[1]{}
\fi

\title{Cross-Paradigm Models of Restricted Syndrome Decoding with Application to CROSS\thanks{This work was funded by the Luxembourg National Research Fund (FNR) via project PQseal C24/IS/18978392.}}

\usepackage{xcolor}

\begin{document}
\maketitle

\begin{abstract}
Restricted Syndrome Decoding (ResSD) is a variant of linear code decoding problem where each of the error's entries must belong to a fixed small set of values. This problem underlies the security of CROSS, a post-quantum signature scheme that is one of the Round~2 candidates of NIST's ongoing additional signatures call. We show that solutions to this problem can be deduced from vectors of a particular structure and a small norm in newly constructed codes, in both Hamming and Euclidean metrics. This allows us to reduce Restricted Syndrome Decoding to both code-based (Regular Syndrome Decoding) and lattice-based problems (Closest Vector Problem, List of Short/Close Vectors), increasing the attack surface and providing new insights into the security of ResSD. We evaluate our attacks on CROSS instances both theoretically and experimentally on reduced parameters.
\end{abstract}

\if\Draft1
\setcounter{tocdepth}{4}
\begingroup
\let\clearpage\relax
\tableofcontents
\endgroup
\fi

\section{Introduction}
\label{sec:intro}

\TODO{Check consistent ref/autoref}
\TODO{Search pdf for ? (missing references/ citations), before submission}

Since the discovery of Shor's algorithm \cite{S94a}, the public key cryptosystems that were and are widely used can be broken by a quantum computer. That is why the US state agency NIST launched in 2017 a competitive call for designing new quantum-resistant cryptosystems, which resulted  in three post-quantum standards in August 2024. While one signature is based on hash functions \cite{NISTPQC:SPHINCS+22}, the two others (one encryption and one signature) assume the hardness of lattice-based problems \cite{NISTPQC:CRYSTALS-KYBER22,NISTPQC:CRYSTALS-DILITHIUM22}. There is also another lattice-based signature standard \cite{NISTPQC:FALCON22} which should be coming soon. This led the NIST to diversify post-quantum standards in order to avoid relying on the same kind of problem. This was the aim of the new NIST call for additional signatures in 2023, whose process is currently in the second round.\footnote{\url{https://csrc.nist.gov/Projects/pqc-dig-sig/round-2-additional-signatures}} For encryption schemes, the NIST announced in March 2025 a new post-quantum standard \cite{NISTPQC-R4:HQC22} based on linear codes.

Indeed, code-based problems are strong candidates for building alternative schemes. One of the remaining candidates for the ongoing signature call is CROSS \cite{NISTPQC-ADD-R2:CROSS24} which is based on a variant of the syndrome decoding problem, the cornerstone of code-based cryptography. It consists in solving an underdetermined linear system of equations over a finite field, but with the additional constraint that the number of nonzero variables is bounded. The idea of using this NP-hard problem (\cite{BMT78}) in cryptography is not new, since the first code-based encryption scheme, the McEliece scheme \cite{M78}, was published in 1978. The syndrome decoding problem has been studied a lot for a long time \cite{P62,S88,D91,B97,J01,FS09,BLP11,BJMM12,MO15,BM17,DT17a,BM18,SAC:BCDL19,KL24,EC:LWYY24}, and this is a testimony to its hardness. One of its variants is the \emph{regular syndrome decoding problem}, where the solution is partitioned into consecutive blocks of equal length, with each block containing exactly one nonzero element. Originally introduced in \cite{AFS05}, it has since been the subject of several particular studies \cite{S07,AC:HOSS18,EC:CarCouJou23,EC:BriOyg23,C:EssSan24,KL24,WWYLYZW25}.

CROSS is based on another variation of the problem called the \emph{restricted syndrome decoding problem}. In this case, each of the solution's entries must belong to a fixed small set of values. First introduced in \cite{ARXIV:BBCHTPSW21,AMC:BBCHPSW25} and proved NP-complete in \cite{ARXIV:BBCHTPSW21,ISIT:BBPSWW23}, it has since been studied in \cite{ISIT:BBPSWW23,PKC:BBPSWW24,CiC:BeuBriOyg24} and by CROSS designers themselves \cite{CROSSsec}. One can see that this problem is relatively new, and that is why NIST is looking forward to further analysis of the underlying problems of CROSS \cite{NISTPQC-ADD-R2}, as well as the ZKPoK identification protocol \cite{PKC:BBPSWW24} used by the scheme.

\subsection*{Our contribution}

Our primary goal is to better understand the Restricted Syndrome Decoding (ResSD) problem and expand the surface of analysis techniques applicable to it. Given a parity-check matrix $\Hm\in\fq^{(n-k)\times n}$ and a syndrome $\sv\in\fq^{n-k}$, the problem consists in finding an error $\ev\in\fq^n$ such that $\ev\transp{\Hm} = \sv$ with each entry of $\ev$ belonging to a fixed subset $E$ of $\fq$. Although CROSS has also a version using ResSD \emph{with a subgroup of $E^n$}, in this work we focus on the basic version (called CROSS-R-SDP). The main family of algorithms for solving the original syndrome decoding problem is called ISD (Information Set Decoding), and almost all existing algorithms for solving ResSD are adaptations of ISD. There is an exception in recent work \cite{CiC:BeuBriOyg24} that studies algebraic techniques based on Gr\"obner basis, as well as hybrid attacks. We depart from these methods with several new ideas.

\medskip

\emph{(Reduction of ResSD to regular syndrome decoding)}
The first idea consists, instead of using $\Hm$ directly, in considering the new parity check matrix
\[\Hm'\ =\ \begin{pmatrix}
    r_1\Hm_1 & \dots & r_z\Hm_1 \ \Vert \ \dots & \dots\ \Vert \ r_1\Hm_n & \dots & r_z\Hm_n
\end{pmatrix}\]
where $\Hm_i$ is the $i^{\text{th}}$ column of $\Hm$ and $r_i$ are the elements of $E$. One can immediately see that if $\ev$ was a solution to our original ResSD problem, then if we denote $\ev'$ the vector of length $zn$ which in the $i^{\text{th}}$ $z$-size block has exactly one~1 at the position that corresponds to the value of the $i^{\text{th}}$ coordinate of $\ev$, we will have $\ev'\Hm' = \sv$. So $\ev'$ is a special solution of a regular syndrome decoding instance, because its nonzero values contain only ones. We call this kind of vectors \emph{light-regular} vectors. That is why we add $n$ parity-check equations imposing a sum equal to 1 on each block (similarly to the so-called ``encoding regularity'' from \cite{C:EssSan24}). We obtain a new parity-check matrix $\widetilde{\Hm}\in\fq^{(2n-k)\times zn}$. We adapt some ISD algorithms to this new particular problem, obtaining rather high complexities because of the high rate of the code admitting $\widetilde{\Hm}$ as a parity-check matrix.

\emph{(Reduction of light-regular decoding (and ResSD) to CVP)}
Noticing that light-regular vectors also minimize the Euclidean norm among solutions, we create the lattice containing this code and convert the regular syndrome decoding problem into a closest vector problem. The idea is similar to lattice-based attacks against code-based problems using the Lee metric \cite{EPRINT:HKNSC24}. But the large dimension of our lattice which in addition has many short vectors makes this instance hard to solve. We nevertheless improve the attack by using a ``hybrid'' method consisting in ``guessing'' a certain number of blocks. This reduces the dimension of the lattice, and if we enumerate all vectors close to our new target, the solution will be among them. We combine this already hybrid attack with a ``truncation method'' consisting in guessing of a certain number of coordinates per block, giving better complexities.

\emph{(Direct reductions from ResSD to ListCVP/ListSVP)}
Another idea is based on the fact that given $a,b\in\fq$ with $a\ne0$ and $\bv = (b,\dots,b)$,  we can rewrite $\ev\transp{\Hm} = \sv$ into $(a\ev + \bv)\transp{\Hm} = a\sv+\bv\transp{\Hm}$. If we denote $\Hm' = a\Hm$ and $\sv' = a\sv+\bv\transp{\Hm}$, they both constitute an instance of ResSD with solution's entries lying in $aE+b$. Thus, we can perform an affine substitution of the set of error's values to minimize a parameter that we called the affine diameter of $E$. This, combined with the truncation method, leads to solutions of small Euclidean norms that allow us to use lattice-based attacks again, while preserving the dimension $n$ of the problem. This results in a reduction of ResSD to a list-closest-vector-problem (which consists in listing the vector close up to a certain distance to the target) and to the list-shortest-vector-problem (the same, but the target is the zero vector) over lattices. Remarkably, certain parameters of ResSD instances heuristically lead to a degeneration of these list-decoding problems into basic CVP and SVP instances. This generalizes the observation of the CROSS designers \cite{CROSSsec} that ResSD with a two-element set $E$ corresponds to a subset-sum problem.

These attacks are based on statistics on the error distributions. In CROSS, the set $E$ of values for the error is a multiplicative subgroup of $\fq$, which allows us to randomize the solution of the ResSD instance. This allows us to convert weak-key attacks into standard probabilistic attacks.

\begin{table}[p!]
    \centering
    \setlength{\tabcolsep}{0.4em}
    \begin{tabular}{cp{1.4cm}cp{4.6cm}c}
    \toprule
    Field & Restr. $E$ & Problem & Parameters & Ref. \\
    \toprule
    $\F_q$ & Any & Regular SD & $n$ blocks of $z$ entries each, enforced binary solution, deterministic & \autoref{sec:regular}\\
    $\F_p$ & Any & CVP & full-rank lattice of dimension $zn$, volume $p^{2n-k}$, norm $=\sqrt{n}$, enforced binary solution, deterministic & \autoref{sec:cvp}\\
    $\F_p$ & Any & ListCVP/CVP & full-rank lattice of dimension $n$, volume $p^{n-k}$, norm $\le\sqrt{n}\cdot D_E/2$, deterministic & \autoref{sec:listcvp-aff-diam}\\
    $\F_p$ & Mult. subgroup & ListCVP/CVP & full-rank lattice of dimension $n$, volume $p^{n-k}$, norm $\le\sqrt{n}\cdot\sigma(E)$, probabilistic & \autoref{sec:listcvp-mean-median}\\
    $\F_p$ & Mult. subgroup & ListSVP/SVP & full-rank lattice of dimension $n$, volume $p^{n-k-1}$, norm $\le\sqrt{n}\cdot\sigma(E)$, probabilistic & \autoref{sec:listcvp-svp} \\
    \bottomrule
    \end{tabular}
    \vspace{0.25em}
    \caption{Summary of reductions from a ResSD problem with parity check matrix $\Hm \in \F_q^{(n-k)\times n}$ and syndrome $\sv \in \F_q^{n-k}$ with solution entries in a set $E$ of size $z$. 
    \ListCVP/\ListSVP{} instances may degrade to \CVP/\SVP{} for some parameters, under Gaussian Heuristic.
    $\sigma(E)$ denotes standard deviation of a random variable sampled uniformly from $E$. $D_E$ denotes the difference between the maximum and the minimum elements of $E$, as integers, up to affine substitution.
    The two last reductions are weak-key attacks for arbitrary restriction sets $E$ and can be converted into general probabilistic attacks when $E$ is a multiplicative subgroup of the field.
    }
    \label{tab:reductions}
\end{table}

\begin{table}[p!]
    \centering
    \setlength{\tabcolsep}{0.5em}
    \begin{tabular}{clccc}
    \toprule
    $(n,k)$ & Attack & Time & Memory & Ref.\\
    \toprule
    \multirow{5}{*}{(127,76)}
    & Shifted representations
    & $2^{143}$ & $2^{117}$ & \cite[Tab.5]{NISTPQC-ADD-R2:CROSS24} \\
    & Naive ISD enumeration
    & $2^{220}$ & poly($n$) & \autoref{st:basic-attack} \\
    & ISD for Regular SD
    & $2^{197}$ & $2^{162}$ & \autoref{subsec:RegAnalisys} \\
    & Hybrid-BatchCVP
    & $2^{229}$ & $2^{136}$ & \autoref{sec:security-batchcvp} \\
    & Hybrid-ListCVP
    & $2^{196}$ & $2^{26}$ & \autoref{sec:security-listcvp} \\
    \midrule
    \multirow{5}{*}{(187,111)}
    & Shifted representations
    & $2^{207}$ & $2^{169}$ & \cite[Tab.5]{NISTPQC-ADD-R2:CROSS24} \\
    & Naive ISD enumeration
    & $2^{319}$ & poly($n$) & \autoref{st:basic-attack} \\
    & ISD for Regular SD
    & $2^{285}$ & $2^{238}$ & \autoref{subsec:RegAnalisys} \\
    & Hybrid-BatchCVP
    & $2^{336}$ & $2^{201}$ & \autoref{sec:security-batchcvp} \\
    & Hybrid-ListCVP
    & $2^{287}$ & $2^{39}$ & \autoref{sec:security-listcvp} \\
    \midrule
    \multirow{5}{*}{(251,150)}
    & Shifted representations
    & $2^{274}$ & $2^{234}$ & \cite[Tab.5]{NISTPQC-ADD-R2:CROSS24} \\
    & Naive ISD enumeration
    & $2^{429}$ & poly($n$) & \autoref{st:basic-attack} \\
    & ISD for Regular SD 
    & $2^{380}$ & $2^{320}$ & \autoref{subsec:RegAnalisys} \\
    & Hybrid-BatchCVP
    & $2^{450}$ & $2^{271}$ & \autoref{sec:security-batchcvp} \\
    & Hybrid-ListCVP
    & $2^{384}$ & $2^{52}$ & \autoref{sec:security-listcvp} \\
    \bottomrule
    \end{tabular}
    \vspace{0.25em}
    \caption{Summary of existing and our attacks against CROSS-R-SDP instances.}
    \label{tab:attacks}
\end{table}

\medskip

Although CROSS insofar is the only scheme based on ResSD, our analysis would help with other potential future designs. We summarize our reductions in  \autoref{tab:reductions} and give all the complexities obtained for the CROSS parameters in \autoref{tab:attacks}. We see that the direct adaptation of ISD algorithms given in \cite{NISTPQC-ADD-R2:CROSS24} remains the best attack, and thus our reductions together with basic RegSD/Lattice-based solving techniques do not threaten CROSS so far. We nevertheless obtain an interesting Time/Memory compromise using our reduction to List-CVP.

\medskip

The relevant source code and supporting materials are available at \cite{ZenodoCodeRSD}.\footnote{
See also
\url{https://github.com/cryptolu/Cross-Paradigm-RestrictedSD}.
}


\section{Preliminaries}
\label{sec:prelim}

The finite field of size $q=p^k$, $p$ prime, is denoted by $\F_q$. The symbol $\eqdef$ is used to define the object on the left side. We denote by $\card{S}$ the cardinality of the set $S$. Vectors are denoted by bold lowercase letters and indexed from~1, matrices are denoted by bold capitals. 
For every matrix $\Mm$, we denote by $\Mm_i$ its $i^{\text{th}}$ column and by $\Mm_{i,j}$ its entry in the $i^{\text{th}}$ row and $j^{\text{th}}$ column. $\Imat_n$ and $\Omat$ are, respectively, the identity matrix of size $n$ for every integer $n$ and the zero matrix ($\Omat$ can also be the zero vector, depending on the context). 
The $i^{\text{th}}$ entry of a vector $\xv$ is denoted $x_i$. 
The \emph{Hamming weight} of a vector $\uv\in\F_q^n$ is $$\hnorm{\uv}\eqdef\pset{i\in\pset{1\dots n}\ \vert\ u_i\ne 0}.$$ 
The \emph{Euclidean norm} of a vector $\vv\in\Z^n$ is
\[
\enorm{\vv} \eqdef \sqrt{\sum_{i=1}^nv_i^2} \,.
\]

\subsection{Linear codes}

\begin{definition}[Linear code]
	A linear code $\CC\subset{\fq^n}$ of dimension $k$ is a vector subspace of $\fq^n$ of dimension $k$. It can be represented by two equivalent ways:
	\begin{itemize}
		\item With a generator matrix $\Gm\in\fq^{k\times n}$ whose rows are a basis of $\CC$: \[\CC = \left\{\xv\Gm\;\vert\;\xv\in\fq^k\right\}\]
		\item With a parity-check matrix $\Hm\in\fq^{(n-k)\times n}$ that is a generator matrix of $\CC^{\perp}$: \[\CC = \left\{\yv\in\fq^n\;\vert\;\yv\Hm^\mathsf{T} = \Omat\right\}\]
	\end{itemize}
	In particular, we have $\Gm\Hm^\mathsf{T} = \Omat$ and $\CC$ does not change if its generator or parity-check matrix is multiplied on the left by an invertible matrix.
\end{definition}

\begin{notation}
	$\CC(\Gm)$ is the linear code that admits $\Gm$ as a generator matrix, and $\CC^\perp(\Hm)$ is the code that admits $\Hm$ as a parity-check matrix.
\end{notation}
We say that a generator matrix $\Gm$ is in \textit{systematic form} if it is in the form $(\Imat_k,\Rm)$.


\subsubsection{Restricted syndrome decoding}

There are several variations of Restricted Syndrome Decoding problems (\ResSD{}) in the literature. Since we focus on the analysis of CROSS, we use the definition given in the specification, which we slightly generalize to include derived instances that we will analyze.

\begin{problem}[\ResSD{}]
Given $\Hm \in \F_q^{(n-k)\times n}$ and $\vec{s} \in \F_q^{n-k}$, find $\vec{e} \in E^n$ such that $\vec{e}\transp{\Hm} = \vec{s}$, where $E \subset \F_q$ is a predetermined subset of size $z$.
\end{problem}
In the CROSS specification \cite{NISTPQC-ADD-R2:CROSS24} and in \cite{CiC:BeuBriOyg24}, the set $E$ is chosen to be a multiplicative subgroup of $\F_q$, and $q=p$ is a prime. We provide the general definition as in \cite{PKC:BBPSWW24}, since we will also consider ``truncated'' instances where $E$ is a subset of the original multiplicative subgroup. We also focus on the full-weight instance, although \cite{ISIT:BBPSWW23,PKC:BBPSWW24} further considered a case where the number of coordinates $e_i = 0$ is fixed.

The following two statements show basic manipulation methods for \ResSD{} instances and are a formalization of the ``Shifting E'' paragraph from \cite{NISTPQC-ADD-R2:CROSS24}.

\begin{proposition}[Affine Shifting of $E$]\label{st:affine-change}
Let $(\Hm,\sv)\in\F_q^{(n-k)\times n}\times\F_q^{n-k}$ be an instance of \ResSD{} with solution's entries in a set $E=\{r_1,\ldots,r_z\}$. We denote by $\uv = (1,\dots,1)$ the vector of length $n-k$. For any $a,b \in \F_q$, $a \ne 0$, define a new instance $(\Hm',\sv')\in\F_q^{(n-k)\times n}\times\F_q^{n-k}$
of \ResSD{} with solution's entries in the set $E'$, where $\Hm'=a^{-1}\Hm$, $\sv'=\sv + ba^{-1}\uv\transp{\Hm}$, $E'=aE+b \subseteq \F_q$.
If there exists an algorithm solving $\left(\Hm',\sv'\right)$ over $E'$ in time $T$ with success probability $\rho$, then there exists an algorithm solving $(\Hm,\sv)$ over $E$ in time $T + \mathcal{O}(n^2)$ with success probability $\rho$.
\end{proposition}
\begin{proof}
Computing $(\Hm',\sv')$ can be done in time $\mathcal{O}(n^2)$. Let us then suppose that we know $\ev'\in E'^n$ such that
\begin{align*}
    \ev'\transp{\Hm'}=\sv' & \Leftrightarrow \quad\ \ev'(a^{-1}\transp{\Hm}) = \sv + ba^{-1}\uv\transp{\Hm}\\
                           & \Leftrightarrow\quad a^{-1}(\ev'-b\uv)\transp{\Hm} = \sv
\end{align*}
We also remark $\ev\eqdef a^{-1}(\ev'-b\uv)$ has entries in \[a^{-1}E'-a^{-1}b = a^{-1}(aE+b)-a^{-1}b = E\]
So $\ev$ that can be computed in linear time is a solution to our problem.\qed
\end{proof}

\begin{proposition}[Multiplicative Randomization]\label{st:mult-randomization}
Let $(\Hm,\sv)\in\F_q^{(n-k)\times n}\times\F_q^{n-k}$ be an instance of \ResSD{} with
restriction set
$E = \pset{\omega^i \mid 0 \le i < z}$ for a primitive $z^{\text{th}}$ root of unity $\omega\in \F_q$. Then, one can sample in polynomial time a new instance $(\Hm',\sv)\in\F_q^{(n-k)\times n}\times\F_q^{n-k}$ of \ResSD{} over the same restriction set $E$, such that the sets of solutions to the two instances are in bijection computable in $\mathcal{O}(n)$ and the image of any solution to $(\Hm,\sv)$ is uniformly distributed over $E^n$.
\end{proposition}
\begin{proof}
Sample a vector $\cv$ uniformly at random from $E^n$ and let $\Cm$ be the diagonal matrix constructed from $\cv$. Observe that a solution $\ev$ to the initial problem satisfies $$\sv=\ev\Hm^\mathsf{T} = \ev \Cm \Cm^{-1} \Hm^\mathsf{T} = \ev' \Hm^{\prime \mathsf{T}},$$
where $\Hm' = \Hm \Cm^{-1}$, and $\ev'=\xv \Cm$ has every coordinate multiplied by a random element from $E$. Observing that because $\omega$ is a primitive root, for any $e_0\in E$ the mapping
$e \mapsto e_0e$ is a bijection from $E$ to itself, this concludes the proof.\qed
\end{proof}
The following method is inspired by hybrid methods in regular syndrome decoding (see below and in \autoref{sec:regular}). It further exploits the multiplicative subgroup structure to provide stronger attack guaranties.

\begin{proposition}[Multiplicative Truncation]
\label{st:mult-truncation}
Let $(\Hm,\sv)\in\F_q^{(n-k)\times n}\times\F_q^{n-k}$ be an instance of \ResSD{} with an existing solution $\xv$ whose entries lie in $E = \pset{\omega^i \mid 0 \le i < z}$ for a primitive $z^{\text{th}}$ root of unity $\omega\in \F_q$. Then, for any integer $z'$ with $1 \le z' \le z$, one can efficiently sample a \ResSD{} instance $(\Hm',\sv)\in\F_q^{(n-k)\times n}\times\F_q^{n-k}$ such that, with probability at least $\rho$ (over sampling), the entries of the solution will lie in $E' = \pset{\omega^i \mid 0 \le i < z'}$, where $
\rho = \pround{z'/z}^n
$,
and a solution to $(\Hm',\sv)$ in $E'^n$ can be efficiently converted into a solution to $(\Hm,\sv)$ in $E$.
\end{proposition}
\begin{proof}
Let $(\Hm',\sv)$ be sampled using multiplicative randomization (\autoref{st:mult-randomization}) and let $\xv'$ be the image of $\xv$ under the bijection from the statement.
Clearly, every coordinate of $\xv'$ lands in $E'$ with probability $z'/z$, independently of other coordinates. Therefore, $\xv'$ belongs to $E^{\prime n}$ with probability $\rho$. Recovering $\xv$ from any solution $\xv'$ amounts to $n$ multiplications (by the diagonal of $\Cm^{-1}$); the result clearly belongs to $E^n$ and satisfies the parity check equation.\qed
\end{proof}
\begin{remark}
This proposition is heavily based on the fact that $E$ is a multiplicative subgroup. For general $E$, it is possible to apply a variant of the truncation similar to hybrid attacks on the regular syndrome decoding problem described in \cite{EC:BriOyg23,WWYLYZW25}. Concretely, we do not modify the parity check matrix and simply guess $z'$ out of $z$ possible values of each error coordinate separately. This has the same probability $1-\binom{z-1}{z'}/\binom{z}{z'}=z'/z$ per coordinate. However, the key difference from the multiplicative truncation is that one cannot force a particular subset of $E$, and every coordinate will have its own randomly chosen subset.
\end{remark}
Finally, we also describe a basic memoryless attack that will be used for comparison with other low-memory methods. Note that the CROSS security analysis document \cite{CROSSsec} only describes attacks with the best time complexity, which have impractical memory complexities.

\begin{proposition}
\label{st:basic-attack}
Let $(\Hm,\sv)\in\F_q^{(n-k)\times n}\times\F_q^{n-k}$ be an instance of \ResSD{} with an existing solution whose entries lie in a set $E$ with $|E|=z$. Then, it can be solved in time $\mathcal{O}(n^2z^k)$ and negligible memory.
\end{proposition}
\begin{proof}
Write $\Hm = (\Imat_{n-k}, \Hm')$ so that for $\ev=(\xv,\yv) \in \F_q^{n-k}\times \F_q^k$, $\Hm \ev^T= \xv^T + \Hm'\yv^T = \sv^T$. Enumerate $z^k$ candidates for $\yv$ and verify that $\xv^T=\sv^T-\Hm'\yv^T$ belongs to $E^{n-k}$ (which can be done in $\mathcal{O}(n^2)$ field operations per guess).
\qed
\end{proof}


\subsubsection{Regular syndrome decoding}

We recall the Regular Syndrome Decoding problem (\RegSD{}).

\begin{definition}
    For $n = z t$, a vector $\ev\in\F_q^n$ of Hamming weight $t$ is \emph{regular} if it can be written $\ev = (\ev_1,\dots,\ev_t)$ where each $\ev_i \in \F_q^z$ has Hamming weight~1.
\end{definition}

\begin{problem}[\RegSD{}]
Given $\Hm \in \F_q^{(n-k)\times n}$ and
$\vec{s} \in \F_q^{n-k}$, $n=zt$, find a regular vector $\vec{e} \in \F_q^n$ of Hamming weight $t$ such that $\vec{e}\transp{\Hm} = \vec{s}$.
\end{problem}

\begin{definition}
    In this work, we will often consider regular vectors over $\F_q^n$ having only one value on their nonzero entries. We are going to call such vectors \emph{light-regular vectors} and denote by $L_t^n$ the set of those of Hamming weight $t$.
\end{definition}


\subsection{Lattices}

An $n$-dimensional ball of radius $R$, denoted $\Ball{n,R} \subseteq \R^n$, consists of all vectors of $\R^n$ having the Euclidean norm at most $R$. Its volume is \[
\Vol{\Ball{n,R}} \eqdef \frac{R^n\pi^{n/2}}{\Gamma(1+n/2)} = R^n \cdot \Vol{\Ball{n,1}}.
\]

\begin{definition}[Lattice]
Let $\Am\in\R^{k\times n}$. The $k$-dimensional \emph{lattice} generated by $\Am$ is $\LL(\Am) \eqdef \left\{\xv\Am\ \vert\ \xv\in\Z^k\right\}$. The length of its shortest nonzero vector (in terms of the Euclidean norm) is denoted $\lambda_1(\LL(\Am))$.
The \emph{volume} of $\Am$ is defined as $\Vol{\LL(\Am)}\eqdef\sqrt{|\det{\Am\Am^T}|}$.
\end{definition}
We now describe fundamental lattice-related problems.
The following standard heuristic is used to approximate the number of lattice vectors of norm below a given radius $r$.

\begin{heuristic}[Gaussian Heuristic (GH)]
For a lattice $\LL \subset \R^n$, the number of lattice points inside a ball $\Ball{n,R}$ is approximately \[
|\LL \cap \Ball{n,R}| \approx \frac{\Vol{\Ball{n,R}}}{\Vol{\LL}}.
\]
\end{heuristic}
In particular,
under GH, the shortest vector is approximated as \[
\lambda_1(\LL)
~\approx~
\pround{\Vol{\Ball{n,1}}}^{-1/n} \pround{\Vol{\LL}}^{1/n} ~\approx~
\sqrt{\frac{n}{2\pi e}} \pround{\Vol{\LL}}^{1/n}.
\]
Below, we describe standard lattice problems and their list-decoding variants and summarize the best solving complexities based on standard heuristics. All mentioned algorithms are probabilistic.

\begin{problem}[Shortest Vector Problem, \SVP]
Given a lattice $\LL\subset\R^n$, find a nonzero vector $\uv$ of $\LL$ of minimum Euclidean norm $\lambda_1(\LL)$.
\end{problem}

\begin{problem}[\ListSVP]
Given a lattice $\LL\subset\R^n$ and a norm upper bound $r$, list all vectors of $\LL$ of norm at most $R$. Equivalently, compute $\LL \cap \Ball{n,R}$.
\end{problem}

\begin{problem}[Closest Vector Problem, \CVP]
Given a lattice $\LL\subset\R^n$ and a vector $\yv\in\R^n$, find a vector $\uv\in\LL$ that minimizes $\enorm{\yv-\uv}$.
\end{problem}

\begin{problem}[\ListCVP]
Given a lattice $\LL\subset\R^n$, a vector $\yv\in\R^n$, and a norm upper bound $R$, list all vectors of $\LL$ having distance at most $R$ to $\yv$. Equivalently, compute $\LL \cap (\yv+\Ball{n,R})$.
\end{problem}

\begin{heuristic}[Sieving \cite{SODA:BDGL16}]
\label{st:sieving}
The SVP problem over an $n$-dimensional lattice can be solved in time $2^{0.292n+o(n)}$ and memory $2^{0.208n+o(n)}$.
\end{heuristic}

Single \CVP{} complexity is the same as sieving; however, there is a speedup when multiple targets need to be solved.

\begin{heuristic}[{Batch-CVP Complexity \cite{PKC:DucLaavWo20}}]\label{heur:cvp-cost}
A batch of $M$ CVP instances over an $n$-dimensional lattice, $M\ge2^{0.058n+o(n)}$, can be solved in time $M\cdot 2^{0.234n+o(n)}$ and memory $2^{0.208n+o(n)}$.
\end{heuristic}

The following heuristic claim is based on classic enumeration methods \cite{AMS:FinPoh85}, \cite{MP:SchEuc94,IWCC:HanPujSte11}, which we summarize in \autoref{app:enum}.

\begin{heuristic}[\ListSVP/\ListCVP{} complexity]
\label{st:enum}
The \ListSVP/\ListCVP{} problems over an $n$-dimensional lattice $\LL$, up to a negligible amount of missed vectors, can be solved in time 
\[
2^{0.292n+o(n)} + \mathcal{O}\pround{
n\sum_{i=1}^n \frac{
\Vol{\Ball{i,R_i}}
}{
\delta^{i(n-i)} \pround{\Vol{\LL}}^{i/n}
}
}
\]
and $2^{0.208n+o(n)}$ memory,
where $\delta=\pround{\frac{n}{2\pi e}}^{-\frac{1}{2n-2}}$ and $R_i = R\cdot \sqrt{\frac{i}{n}}$.
\end{heuristic}

We will also use the standard notions of LLL-reduced, BKZ-$\beta$ reduced ($2 \le \beta \le n$) and HKZ-reduced lattice bases. Their formal definition is given in \autoref{app:enum}.



\section{Regular Syndrome Decoding Modeling}
\label{sec:regular}

In this section, $E$ is a fixed subset of $\fq$ of cardinality $z$, whose elements are $\left\{r_1,\dots,r_z\right\}$.

\begin{notation}\label{not:expand}
    For every matrix $\Mm$ over $\fq$ with $n$ columns, we denote  the $zn$-column matrix \[ \Mm^E \eqdef \begin{pmatrix}
	   r_1\Mm_1 & \dots & r_z\Mm_1 \ \Vert \ \dots & \dots\ \Vert \ r_1\Mm_n & \dots & r_z\Mm_n
    \end{pmatrix} \]
\end{notation}

\subsection{Reduction}\label{subsec: reduc}

We consider the following matrix and vector
\[\Um_n = \begin{pmatrix}
	   1\dots 1 & \Vert & \ 0 & \dots & \dots & 0 \\
	   0 & \ddots & & \ddots & & \vdots \\
	   \vdots & \ddots\ \ & & \ddots &\ \ \ddots & \vdots \\
	   \vdots & & \ddots & &\ \ \ \ \ddots & 0 \\
	   0 & \dots & \dots & 0 & \Vert & 1\dots 1
    \end{pmatrix}\in\fq^{n\times nz} \ \ \text{and}\ \  \uv_n = \begin{pmatrix}
1 \dots\dots\dots 1
\end{pmatrix}\in\fq^n\]
for the following definition.

\begin{definition}\label{def:regular}
    For any $\Hm\in\fq^{(n-k)\times n}$ and $\sv\in\fq^{n-k}$, we define 
    \[\widetilde{\Hm}\eqdef\begin{pmatrix}
	   \Um \\ \Hm^E
    \end{pmatrix}\in\F_{p}^{(2n-k)\times zn}\quad and \quad\widetilde{\sv}\eqdef\begin{pmatrix}
        \uv & \vert & \sv
    \end{pmatrix}\in\F_{p}^{2n-k}\]
    For any $\ev\in E^n$, we also define $\widetilde{\ev} = (\widetilde{\ev}_1,\dots,\widetilde{\ev}_n)\in\F_q^{zn}$ such that for $1\leq i\leq n$, if the $i^{\text{th}}$ entry of $\ev$ is equal to $r_j$, then $\widetilde{\ev}_i\in \F_q^z$ has value one at its $j^{\text{th}}$ entry and zeros elsewhere.
\end{definition}
This construction allows the following lemma:

\begin{lemma}\label{lem:ResRegBij}
    Let $(\Hm,\sv)\in\F_q^{(n-k)\times n}\times\F_q^{n-k}$ and $\left(\widetilde{\Hm},\widetilde{\sv}\right)\in\F_q^{(2n-k)\times zn}\times\F_q^{2n-k}$ defined as in \autoref{def:regular}. Every solution $\xv$ of $\xv\transp{\widetilde{\Hm}} = \widetilde{\sv}$ such that $\hnorm{\xv}\leq n$ belongs to $L^{zn}_n$ and the mapping
    \begin{align*}
	   \varphi:E^n &\longrightarrow L_n^{zn}\\
	   \ev\ \ &\longrightarrow\ \widetilde{\ev}
    \end{align*} forms a bijection between the solutions in $E^n$ of the \ResSD{} instance $(\Hm,\sv)$ and the $\RegSD{}$ instance $\left(\widetilde{\Hm},\widetilde{\sv}\right)$ with solutions of Hamming weight $n$.
\end{lemma}

\begin{proof}
    By construction of the $n$ first rows of $\widetilde{\Hm}$ and the $n$ first entries of $\widetilde{\sv}$, $\xv\transp{\widetilde{\Hm}} = \widetilde{\sv}$ implies that $\xv$ has at least one nonzero entry per block, so $\hnorm{\xv}\geq n$. Adding the constraint $\hnorm{\xv}\leq n$ means that $\hnorm{\xv} = n$ and so $\xv$ is regular. Still, by construction of the $n$ first rows of $\widetilde{\Hm}$ and the $n$ first entries of $\widetilde{\sv}$, we deduce that the unique nonzero value in each block can only be $1$. So $\xv\in L^{zn}_n$.\\
    For $\ev$ a solution of $(\Hm,\sv)$, it is easy to see that $\widetilde{\ev}$ is a solution of $\left(\widetilde{\Hm},\widetilde{\sv}\right)$ and, conversely, for $\uv$ a solution of $\left(\widetilde{\Hm},\widetilde{\sv}\right)$ we can easily see that $\varphi^{-1}(\uv)$ is a solution of $(\Hm,\sv)$.
    Because $\varphi$ and $\varphi^{-1}$ are injective functions, this concludes the proof.\qed
\end{proof}
This bijection naturally gives the following reduction.

\begin{theorem}\label{th:red1}
    Let $(\Hm,\sv)\in\F_q^{(n-k)\times n}\times\F_q^{n-k}$ be an instance of \ResSD{} with solution entries in $E = \pset{r_1,\dots r_z}$, and consider $\left(\widetilde{\Hm},\widetilde{\sv}\right)\in\F_q^{(2n-k)\times zn}\times\F_q^{2n-k}$ defined as in \autoref{def:regular} as an instance of \RegSD{}  with solutions of Hamming weight $n$. If there exists an algorithm solving $\left(\widetilde{\Hm},\widetilde{\sv}\right)$ in time $T$ with success probability $\rho$, then there exists an algorithm solving $(\Hm,\sv)$ in time $T + \mathcal{O}(n^2)$ with success probability $\rho$.
\end{theorem}

\begin{proof}
    It follows immediately from \autoref{lem:ResRegBij} and the facts that computing $\left(\widetilde{\Hm},\widetilde{\sv}\right)$ can be done in quadratic time and recovering $\ev\in E^n$ from $\widetilde{\ev}\in L^{zn}_n$ in linear time.
    \qed
\end{proof}

\begin{remark}
    The light-regular solutions of the instance $\left(\widetilde{\Hm},\widetilde{\sv}\right)$ of \RegSD{} are exactly the vectors of minimum Hamming weight among those satisfying $\xv\transp{\widetilde{\Hm}} = \widetilde{\sv}$ (as seen in \autoref{lem:ResRegBij}), so they are also solutions of the usual Syndrome Decoding problem (but with a structured code). They are also solutions of minimum Lee metric as well as the ones of minimum Euclidean norm, which leads to considering them as solutions of a lattice problem as we do in the next section.
\end{remark}


\subsection{Security analysis}\label{subsec:RegAnalisys}

 Most algorithms for solving the usual decoding problem rely on ISD (Information Set Decoding) \;
 \cite{P62,S88,D91,B97,FS09,BLP11,BJMM12,MO15,BM17,BM18,EC:LWYY24}. These have been adapted for the case of regular vectors in \cite{AC:HOSS18,EC:CarCouJou23,C:EssSan24} for the binary case, and in \cite{EC:LWYY24,KL24} for the general $q$-ary case. We are going to adapt the Permutation-Based and Enumeration-Based Regular-ISD from \cite{C:EssSan24}, which are themselves adaptations of \cite{P62} and \cite{FS09} to the error-regular case.

\subsubsection{Bi-Regular Permutations}

 We should first consider the conditions for our matrix $\widetilde{\Hm}$ to be set in systematic form. Let $\Pm\in\fq^{zn\times zn}$ be a permutation matrix and we denote $\left(\widetilde{\Hm}_1,\widetilde{\Hm}_2\right)\eqdef\widetilde{\Hm}\Pm$, where $\widetilde{\Hm}_1$ is a square matrix of size $2n-k$. $\Pm$~must be selected in order to have $\widetilde{\Hm}_1$ invertible. So, because of the $n$ first rows of $\widetilde{\Hm}$, there must be at least one column of each block in $\widetilde{\Hm}_1$, or there will be empty rows. Moreover, there should not be three columns or more from the same block in $\widetilde{\Hm}_1$, or there will be a linear dependency between columns. Indeed, let us consider $z'\leq z $ columns of the same block. By construction of $\widetilde{\Hm}$ there is a linear dependency between them if and only if there exists a nonzero vector $\xv\in\fq^{z'}$ such that 
\[\sum_{i=1}^{z'}x_i = 0\quad\text{and}\quad\sum_{i=1}^{z'}r_ix_i = 0\]
where $r_i$ are some elements of $E$ (which are all distinct). This system admits a nonzero vectorial solution if and only if $z'\geq 3$.\\
 \\
 That is why we introduce the notion of bi-regular permutations which is inspired by regular permutations introduced in \cite{C:EssSan24}. It allows us to have exactly one or two columns of each block in $\widetilde{\Hm}_1$.
 
\begin{definition}\label{def:biregular}
    Let $\ev = (\ev_1,\dots,\ev_n)\in(\fq^z)^n$. For a permutation matrix $\Pm$ let
    \[\ev\Pm = (\ev_1',\dots,\ev_n',\ev_1'',\dots,\ev_n'',\ev_1''',\dots,\ev_n''')\]
    with $\ev_i'$ and $\ev_i''$ are vectors with one entry, and $\ev_i'''\in\fq^{z-2}$. $\Pm$ is \emph{bi-regular} if each $\ev_i'$, $\ev_i''$ and $\ev_i'''$ are formed only by coordinates from $\ev_i$.
\end{definition} 
Since columns from different blocks cannot have linear dependencies (because of the distinct rows of ones in the $n$ first rows), we obtain the following proposition.

\begin{proposition}
    Let $\widetilde{\Hm}\in\fq^{(2n-k)\times zn}$ defined as in \autoref{def:regular} from some $\Hm\in\fq^{(n-k)\times n}$. For $\Pm\in\fq^{zn\times zn}$ a permutation matrix, if $\Pm$ is bi-regular then $\widetilde{\Hm}\Pm$ can be put in systematic form by doing linear operations on its rows.
\end{proposition}
\subsubsection{Permutation-Based ISD}

This algorithm from \cite{C:EssSan24} was a direct adaptation of the original Prange algorithm \cite{P62} to the regular context. Although we are in a non-binary context, we can almost directly take the version of \cite{C:EssSan24}. Given an instance $\left(\widetilde{\Hm},\widetilde{\sv}\right)$ (defined as in \autoref{def:regular}), the algorithm randomly chooses a bi-regular permutation matrix $\Pm$ for computing $\left(\widetilde{\Hm}_1,\widetilde{\Hm}_2\right)\eqdef\widetilde{\Hm}\Pm$, where $\widetilde{\Hm}_1\in\fq^{(2n-k)\times(2n-k)}$ and $\widetilde{\Hm}_2\in\fq^{(2n-k)\times ((z-2)n+k)}$. Then it computes $\ev_1\eqdef \widetilde{\Hm}^{-1}_1\widetilde{\sv}\in\fq^{2n-k}$ and checks whether $\hnorm{\ev_1} = n$. If this is the case, it returns the solution $\widetilde{\ev} = \Pm^{-1}(\ev_1,\Omat)$. If not, the algorithm starts the process again.

\begin{remark}
    If the vector $\widetilde{\ev}$ has Hamming weight $n$, it will also be light-regular since the linear constraints of $\widetilde{\Hm}$ impose the sum on each of the $n$ blocks to be equal to 1.
\end{remark}

\begin{remark}
    In the notation of \autoref{def:biregular}, we impose $(\ev_{n-k+1}'',\dots,\ev_n'')$ and $(\ev_1''',\dots,\ev_n''')$ to be equal to $\Omat$.
\end{remark}
Assuming that the original solution is unique (as expected with CROSS parameters), we can see that the probability for choosing a right bi-regular permutation is exactly the probability of guessing:
\begin{itemize}
    \item for $n-k$ blocks, a set of two entries to which the nonzero value belongs.
    \item for $k$ blocks, the exact position of the nonzero value.
\end{itemize}
We obtain a probability of success that is
\begin{equation}\label{eq:Prange}
    \left(\dfrac{z-1}{\binom{z}{2}}\right)^{n-k}z^{-k} = 2^{-k}\left(\dfrac{2}{z}\right)^n
\end{equation}
The total average number of iteration of our adapted Permutation-Based ISD being the inverse of the probability of success, the complexity of the algorithm is at least in $\mathcal{O}\left(2^{k}\left(\dfrac{z}{2}\right)^n\right)$.

\subsubsection{Enumeration-Based ISD} We are now going to improve the algorithm using the ideas originally coming from \cite{FS09}. We again adapt to our case the adaptation from \cite{C:EssSan24}. This time $\widetilde{\Hm}$ must be put in the \emph{quasi-systematic} form, which is
$\begin{pmatrix}
	   \Imat_{2n-k-\ell} &  \widetilde{\Hm}_1\\ \Omat & \widetilde{\Hm}_2
\end{pmatrix}$ for an integer $\ell\geq 1$, with $\widetilde{\Hm}_1\in\fq^{(2n-k-\ell)\times((z-2)n + k + \ell)}$ and $\widetilde{\Hm}_2\in\fq^{\ell\times((z-2)n+k+\ell)}$. For that purpose, we need to transform $\widetilde{\Hm}$ by permuting its columns so that its $2n-k-\ell$ first columns form a full-rank matrix. The bi-regular permutations as we have defined them are suitable for this.

So, the algorithm starts by transforming $\widetilde{\Hm}$ into $\Qm\widetilde{\Hm}\Pm$ where $\Pm$ is a random bi-regular permutation and $\Qm$ is an invertible matrix giving the quasi-systematic form. Denoting $\ev = (\ev_1,\ev_2)$ and $\Qm\widetilde{\sv} = (\widetilde{\sv}_1,\widetilde{\sv}_2)$, we just need to solve $\ev_2\transp{\widetilde{\Hm}_2} = \widetilde{\sv}_2$, and we recover $\ev_1 = \widetilde{\sv}_1-\ev_2\transp{\widetilde{\Hm}_1}$. We are going to enumerate the vectors $\ev_2$ verifying $\ev_2\transp{\widetilde{\Hm}_2} = \widetilde{\sv}_2$ with some additional conditions on~$\ev_2$.

Because we have used a bi-regular permutation, we can retake the notation of \autoref{def:biregular} and write 
\begin{align*}
    \ev\quad\quad &= \quad\quad\quad\quad(\ev_1',\dots,\ev_n',\ev_1'',\dots,\ev_n'',\ev_1''',\dots,\ev_n''')\\
         (\ev_1,\ev_2)\quad       &= \quad((\ev_1',\dots,\ev_n',\ev_1'',\dots,\ev_{n-k-\ell}''),(\ev_{n-k-\ell +1}'',\dots,\ev_n'',\ev_1''',\dots,\ev_n'''))
\end{align*}
The vectors $\ev_2$ that we enumerate are those with elements $(\ev_{n-k-\ell +1}'',\dots,\ev_n'')$ set to zero and whose elements $(\ev_1''',\dots,\ev_n''')$ have a Hamming weight equal to a parameter~$p$. Let us denote $\yv\eqdef(\ev_1''',\dots,\ev_n''')$.
\paragraph{Enumerating $\yv$}
This vector is divided into $n$ blocks of $z-2$ entries, each with at most one nonzero entry per block (of value 1). We are going to apply a meet-in-the-middle strategy taking into account this structure. We define the set \[\mathcal{R}^n_p\eqdef\pset{\vv = (\vv_1,\dots,\vv_n)\in\left(\fq^{z-2}\right)^n\ \ \Big\vert\ \ 0\leq\hnorm{\vv_i}\leq1,\ \hnorm{\vv} = p,\ v_{i,j}\in\pset{0,1}}\] and write $\yv = (\yv_1,\yv_2)$ where each $\yv_i$ is composed of $n/2$ blocks.\\ Having $\ev_2\transp{\widetilde{\Hm}_2} = \widetilde{\sv}_2$ now means $\yv_1\transp{\Hm_2'} = \widetilde{\sv}_2 - \yv_2\transp{\Hm_2''}$, where $\widetilde{\Hm}_2 = (\Hm_2,\Hm_2',\Hm_2'')$ with $\Hm_2\in\fq^{\ell\times(k+\ell)}$ and $\Hm'$, $\Hm''\in\fq^{\ell\times(z-2)\frac{n}{2}}$.

We construct three lists $L_1$, $L_2$ and $L$ where $L_i = \pset{\yv_i\ \vert\ \yv_i \in\mathcal{R}^{n/2}_{p/2}}$ and $L = \pset{(\yv_1,\yv_2)\in L_1\times L_2\ \vert\ \yv_1\transp{\Hm_2'}+\yv_2\transp{\Hm_2''} = \widetilde{\sv}_2}$. Then for all $\yv\in L$, we choose $\ev_2 = (\Omat,\yv)$ and compute $\ev_1 = \widetilde{\sv}_1-\ev_2\transp{\widetilde{\Hm}_1}$. The algorithm then checks whether $\hnorm{\ev_1} = n-p$. If this is the case, it returns $\widetilde{\ev} = \Pm^{-1}(\ev_1,\ev_2)$ as a solution. If not, the algorithm starts again.\\
\\
The entire algorithm is described in \autoref{alg:EnBased}.

\begin{algorithm}[H]
\caption{Enumeration-Based algorithm}\label{alg:EnBased}
\textbf{Parameters:} $\ell\leq 2n-k$ and $p\leq n$,\\
\textbf{Inputs:} Parity-check matrix $\widetilde{\Hm}\in\fq^{(2n-k)\times zn}$, syndrome $\widetilde{\sv}\in\fq^{2n-k}$\\
\textbf{Output:} Light-regular vector $\widetilde{\ev}\in\fq^{zn}$ of weight $n$ such that $\widetilde{\ev}\transp{\widetilde{\Hm}}$
\begin{algorithmic}[1]
\Loop
    \State Sample a random bi-regular permutation $\Pm$.
    \State $\begin{pmatrix}
	           \Imat_{2n-k-\ell} &  \widetilde{\Hm}_1\\ \Omat & \widetilde{\Hm}_2
            \end{pmatrix} \gets \Qm\widetilde{\Hm}\Pm$, $(\widetilde{\sv}_1,\widetilde{\sv}_2) \gets \Qm\widetilde{\sv}$
     \State $L_i \gets \pset{\yv_i\ \vert\ \yv_i \in\mathcal{R}^{n/2}_{p/2}}$, $i = 1,2$
     \State $L \gets \pset{(\yv_1,\yv_2)\in L_1\times L_2\ \vert\ (\Omat,\yv_1,\yv_2)\transp{\widetilde{\Hm}_2} = \widetilde{\sv}_2}$
     \For {$\yv\in L$}
     \State $\ev_1 = \widetilde{\sv}_1-(\Omat,\yv)\transp{\widetilde{\Hm}_1}$
     \If{$\hnorm{\ev_1} = n-p$}
        \State\Return $\Pm^{-1}(\ev_1,\ev_2)$
     \EndIf
     \EndFor
\EndLoop
\end{algorithmic}
\end{algorithm}

\paragraph{Complexity} We still assume that the solution is unique. Following \cite{C:EssSan24}, we can estimate that the time of enumeration will be approximately in \[\mathcal{T}_0 = n\max(\card{L_1},\card{L_2},\card{L})\] with $\E(\card{L})=\dfrac{\card{L_1\times L_2}}{q^\ell}$.  As in the previous algorithm, we must compute the probability of choosing a good permutation $\Pm$. This corresponds to the probability of having the following:
\begin{itemize}
    \item $\yv_1$ and $\yv_2$ each of Hamming weight $p/2$ which is according to binomial law \Big(for each block the probability of having weight 1 being $\dfrac{z-2}{z}$\Big): \[\mathcal{P}_0 =\binom{n/2}{p/2}^2\left(\dfrac{z-2}{z}\right)^p\left(\dfrac{2}{z}\right)^{n-p}\]
    \item having $(\ev_{n-k-\ell +1}'',\dots,\ev_n'')$ equal to $\Omat$. We denote this probability by $\mathcal{P}$.
\end{itemize}
For computing $\mathcal{P}$, we first suppose that $k+\ell\geq n/2$. We have $\mathcal{P} = \mathcal{P}_1\mathcal{P}_2$ where $\mathcal{P}_1\eqdef\Prob\pset{(\ev_{n-k-\ell +1}'',\dots,\ev_{n/2}'') = \Omat}$ and $\mathcal{P}_2\eqdef\Prob\pset{(\ev_{n/2}'',\dots,\ev_n'') = \Omat}$. For $\mathcal{P}_2$, we know that there are $p/2$ blocks where the nonzero entries are already set, and for the entries from other blocks the probability to be nonzero is $1/2$ because there remain only two positions. So $\mathcal{P}_2 = \left(\dfrac{1}{2}\right)^{n/2-p/2}$. For $\mathcal{P}_1$, we first compute the probability of having $i$ entries from blocks where the nonzero entry has not already been set in $\yv$ (there are $n/2-p/2$ of such blocks):
\[P_i = \dfrac{\binom{k+\ell -n/2}{i}\binom{n-k-\ell}{n/2-p/2-i}}{\binom{n/2}{p/2}}\quad\text{and so}\quad\mathcal{P}_1 = \sum_{i=0}^{k+\ell-n/2}P_i\ 2^{-i}\]
For the case where $k+\ell\leq n/2$, we have \[\mathcal{P} = \sum_{i=0}^{k+\ell}P'_i\ 2^{-i}\quad\text{where}\quad P'_i = \dfrac{\binom{k+\ell}{i}\binom{n/2-k-\ell}{n/2-p/2-i}}{\binom{n/2}{p/2}}\]
We can now estimate the hole complexity by (when $k+\ell\geq n/2$):
\begin{multline*}
    \dfrac{\mathcal{T}_0}{\mathcal{P}_0\mathcal{P}_1\mathcal{P}_2} =
    \dfrac{n2^{n/2-p/2}\max\left(\binom{n/2}{p/2}(z-2)^{p/2},\binom{n/2}{p/2}^2(z-2)^p/q^\ell\right)}{\binom{n/2}{p/2}^2\left(\dfrac{z-2}{z}\right)^p\left(\dfrac{2}{z}\right)^{n-p}\sum_{i=0}^{k+\ell-n/2}\dfrac{\binom{k+\ell -n/2}{i}\binom{n-k-\ell}{n/2-p/2-i}}{\binom{n/2}{p/2}}2^{-i}} \\
    = \dfrac{n2^{n/2-p/2}\max\left((z-2)^{p/2},\binom{n/2}{p/2}(z-2)^p/q^\ell\right)}{\left(\dfrac{z-2}{z}\right)^p\left(\dfrac{2}{z}\right)^{n-p}\sum_{i=0}^{k+\ell-n/2}\binom{k+\ell -n/2}{i}\binom{n-k-\ell}{n/2-p/2-i}2^{-i}}
\end{multline*}
Memory complexity is proportional to $\card{L_1} = \binom{n/2}{p/2}(z-2)^{p/2}$.

\paragraph{Complexities for CROSS} We give in \autoref{tab:cross-ISD} the complexities of our adapted ISD algorithms, giving also the best parameters $p$ and $\ell$ to run the Enumeration-Based ISD. We recall that with the CROSS parameters, our regular code has length $7n$ and dimension $5n+k$. We can notice that the rate of attack's complexity to the security parameter remains roughly the same: 0.65 for 128, 0.67 for 192, and 0.67 for 256. The Enumeration-Based algorithm has a high memory complexity due to the use of lists.

\begin{table}[htbp]
\setlength{\tabcolsep}{0.4em}
\centering
    \begin{tabular}{c|c|c|c|c}
    \toprule
       Security & $(n,k)$ & Permutation-Based costs & Enumeration-Based costs & $(p,\ell)$\\
      \midrule
        128 & (127,76) & 325 & 197\quad (162) & $(99,23)$\\
        192 & (187,111) & 470 & 285\quad (238) & $(144,34)$\\
        256 & (251,150) & 625 & 380\quad  (320) & $(196,46)$ \\
    \bottomrule
    \end{tabular}
\vspace{0.25em}
\caption{Complexities of our adapted ISD for CROSS ($q=127$ and $z = 7$). Costs are given in bits, and memory complexities are in parenthesis.}
\label{tab:cross-ISD}
\end{table}

\subsubsection{Conclusion}
We have seen that adapting the ISD algorithms for our case does not give competitive attacks. The main reason is because the rate of the regular code is $\dfrac{(z-2)n+k}{zn}$, which tends to be high when $z$ is not too small, as in CROSS where $z=7$. This means that one must guess a large number of variables while applying the adapted Prange algorithm or its refinements.
The complexity could probably be reduced by using other ISD variants but not enough to make it competitive.
The fact that it is more effective to apply the adapted ISD directly to the original $\ResSD{}$ problem as was done in the CROSS specification could possibly have been anticipated.
Indeed, we saw in~$\eqref{eq:Prange}$ that our Prange's average number of iterations is $2^{k}\left(\dfrac{z}{2}\right)^n$, which is strictly greater than $z^k$ because $n>k$. And we can notice that if we ``adapt'' Prange's algorithm to the original $\ResSD{}$ problem directly, we will have to ``guess'' $k$ values that can take $z$ values each (this is the attack of \autoref{st:basic-attack}), making the average number of iterations equal to $z^k$. 

There are also algebraic (or hybrid) algorithms for solving $\RegSD{}$ as studied in \cite{EC:BriOyg23,WWYLYZW25}, in both binary and non-binary cases. However, although they are likely to provide better parameters (the algorithm described in \cite{WWYLYZW25} has the lowest complexities in many parameters regimes), it would also probably not be enough to have interesting complexities. Indeed, if we observe the gaps between their complexities and the ones of adapted ISD given in \cite[Tab.2]{WWYLYZW25} for codes of approximately the same lengths as ours, the gain is up to 15 bits.


\section{Lattice-based Modelings: Reduction of ResSD to CVP}
\label{sec:cvp}

From now on, we will work in $\F_p$ with $p$ a prime number, and $E = \pset{r_1,\dots,r_z}$ is a fixed subset of $\F_p$.

For every matrix $\Mm$ or vector $\vv$ over $\F_p$, we denote by $\Z(\Mm)$ (or $\Z(\vv)$) the matrix (or vector) over $\Z$ whose entries have the same values as those in $\Mm$ (or $\vv$), with identification $\F_p \simeq \{0,\ldots,p-1\}$. Conversely, for a matrix $\Nm$ over $\Z$, we denote by $\F_p(\Nm)$ the matrix over $\F_p$ whose entries are the projection of those of $\Nm$ (same for vectors).
For any linear code $\mathcal{\CC}$,
define
$\Z(\mathcal{\CC})\eqdef\pset{\Z(\vv)\ \vert\ \vv\in\mathcal{\CC}}$. We also consider the lattice $\Z(\CC)+p\Z^n$ that contains $\CC$ precisely.

\begin{definition}\label{def:LatCode}
    Given $\CC\eqdef\CC(\Gm)\subset\F_p^n$ a linear code of dimension $k$ with $\Gm = (\Imat_k,\Rm)$ in systematic form, we define
    \[\LL(\CC) \eqdef \LL\left(\Am^\CC\right) \quad \text{with} \quad\Am^\CC \eqdef\begin{pmatrix}
        \Imat_{k}\ & \Z\left(\Rm\right) \\ \Omat\ & p\Imat_{n-k}
    \end{pmatrix}\in\Z^{n\times n} \]
\end{definition}

\begin{remark}
    If $\Gm$ cannot be put in a systematic form  as is and requires multiplication by a permutation matrix $\Pm\in\F_p^{n\times n}$, we instead consider $\CC \eqdef \CC\left(\Gm\Pm\right)$ and define $\LL(\CC)\eqdef\LL\left(\Am^\CC\Pm^{-1}\right)$.
\end{remark}
Then we can easily check the following proposition.

\begin{proposition}\label{prop:lcbij}
    Let $\CC \subset\F_p^n$ be a linear code. The mapping
    \begin{align*}
	   \Z^n &\longrightarrow \F_p^n\\
	   \uv\ &\longrightarrow \F_p(\uv)
    \end{align*} forms a surjection from $\LL(\CC)$ onto $\CC$, and a bijection between $\LL(\CC)\cap\pset{0,\dots, p-1}^n$ and~$\CC$.    
\end{proposition}
From the structure of the generator matrix $\Am^\CC$ one can immediately deduce the volume of the lattice, equal to the determinant of the matrix.

\begin{proposition}\label{prop:vol}
    For any linear code $\CC\subset\F_p^n$ of dimension $k$, $\Vol{\LL(\CC)}=p^{n-k}$.
\end{proposition}


\subsection{Reduction}

Let $(\Hm,\sv)\in\F_p^{(n-k)\times n}\times\F_p^{n-k}$, and $\left(\widetilde{\Hm},\widetilde{\sv}\right)\in\F_q^{(2n-k)\times zn}\times\F_q^{2n-k}$ which is defined as in \autoref{def:regular}. We denote $\LL^*\eqdef\LL\left(\widetilde{\CC}\right)$ where $\widetilde{\CC}\eqdef\CC^\perp\left(\widetilde{\Hm}\right)$. Let $\widetilde{\yv}$ be some vector of $\F_p^n$ such that $\widetilde{\yv}\transp{\widetilde{\Hm}} = \widetilde{\sv}$ and $\yv^*\eqdef \Z(\widetilde{\yv})$. Let $S^*$ be the set of the solutions of \CVP{} for the instance $(\LL^*,\yv^*)$, and $\widetilde{S}$ be the set of solutions of \RegSD{} of Hamming weight $n$ for the instance $\left(\widetilde{\Hm},\widetilde{\sv}\right)$.

For the reduction, we will firstly need the two following lemmas giving relations between our code and our lattice.

\begin{lemma}\label{lem:EuclToHamm}
    Let $\uv^*\in\LL^*$ be such that $\enorm{\yv^*-\uv^*} \leq\sqrt{n}$. Then $\widetilde{\ev}\eqdef\F_p(\yv^*-\uv^*)$ belongs to $\widetilde{S}$.
\end{lemma}

\begin{proof}
    Let us denote $\widetilde{\uv} = \F_p(\uv^*)$. We have $\widetilde{\ev}\transp{\widetilde{\Hm}} = \widetilde{\yv}\transp{\widetilde{\Hm}} - \widetilde{\uv}\transp{\widetilde{\Hm}} = \widetilde{\sv}$ because $\widetilde{\uv}\in\CC^\perp\left(\widetilde{\Hm}\right)$ thanks to \autoref{prop:lcbij}. If we suppose that $\hnorm{\widetilde{\ev}} > n$, then we will obviously have $\enorm{\yv^* - \uv^*}>\sqrt{n}$ which is not possible. So $\hnorm{\widetilde{\ev}} \leq n$ and by \autoref{lem:ResRegBij} we can deduce that $\widetilde{\ev}\in L_n^{zn}$, so $\widetilde{\ev}\in \widetilde{S}$.\qed
\end{proof}

\begin{lemma}\label{lem:HammToEucl}
    Let $\widetilde{\ev}\in\widetilde{S}$ and we denote $\ev^*\eqdef\Z(\widetilde{\ev})$. Then $\enorm{\ev^*} = \sqrt{n}$ and $\uv^* \eqdef \yv^*-\ev^*$ belongs to~$\LL^*$.
\end{lemma}

\begin{proof}
    By \autoref{lem:ResRegBij} we know that $\widetilde{\ev}\in L^{zn}_n$, so $\enorm{\ev^*} = \sqrt{n}$. Let us denote $\widetilde{\uv} = \F_p(\uv^*)$.  We can check that $\widetilde{\uv}\transp{\widetilde{\Hm}} = \widetilde{\yv}\transp{\widetilde{\Hm}} - \widetilde{\ev}\transp{\widetilde{\Hm}} = \widetilde{\sv} - \widetilde{\sv} = \Omat$. So we have $\widetilde{\uv}\in\CC^\perp\left(\widetilde{\Hm}\right)$ and by \autoref{prop:lcbij} we can deduce that $\uv^*\in\LL^*$.\qed
\end{proof}
We now compute the distance of a solution of our \CVP{} instance to the target.

\begin{lemma}\label{lem:closestNorm}
    If we assume that $\widetilde{S}\ne\emptyset$, then we have \[S^*=\pset{\uv^*\in\LL^*\ \vert\ \enorm{\yv^*-\uv^*} = \sqrt{n}}\]
\end{lemma}

\begin{proof}
    We assumed $\widetilde{S}\ne\emptyset$, so we can choose $\widetilde{\ev}_1\in\widetilde{S}$. We also define $\uv^*_1 = \yv^*-\ev^*_1$ and $\widetilde{\uv}_1\eqdef\F_p(\uv^*_1)$. Thanks to \autoref{lem:HammToEucl} we know that $\enorm{\ev^*_1} = \sqrt{n}$ and $\uv^*_1\in\LL^*$. So we have obtained that $\uv^*_1\in\pset{\uv^*\in\LL^*\ \vert\ \enorm{\yv^*-\uv^*} = \sqrt{n}}$. This means by the definition of \CVP{} that
    \begin{equation}\label{eq:S*}
        S^*=\pset{\uv^*\in\LL^*\ \vert\ \enorm{\yv^*-\uv^*} = R}\quad\text{for\  some}\quad 0\leq R \leq \sqrt{n}.
    \end{equation}
    We now assume that $R < \sqrt{n}$.
    Let $\uv^*_2\in S^*$. We denote $\ev^*_2\eqdef \yv^*-\uv^*_2$, with $\widetilde{\ev}_2\eqdef\F_p(\ev^*_2)$. We know that $\enorm{\ev^*_2}=R$ because of \eqref{eq:S*} so by \autoref{lem:EuclToHamm} we have $\widetilde{\ev}_2\in \widetilde{S}$. So $\hnorm{\widetilde{\ev}_2} = n$ implying $\enorm{\ev^*_2}\geq\sqrt{n}$ which is a contradiction to the fact that $R<\sqrt{n}$. This means that $R =\sqrt{n}$.\qed
\end{proof}
We can now prove the following proposition, which gives in fact a reduction from \RegSD{} to \CVP{} (with our particular instances).

\begin{proposition}\label{prop:solBij}
    If we assume that $\widetilde{S}\ne\emptyset$, then the mapping \begin{align*}
	   \psi:\Z^{zn} &\longrightarrow\quad \F_p^{zn}\\
	   \uv^* \ &\longrightarrow\ \F_p(\yv^*-\uv^*)
    \end{align*} forms a bijection between $S^*$ and $\widetilde{S}$.
\end{proposition}

\begin{proof}We first show that $\psi$ with the domain restrained to $S^*$ is injective.
    Let us suppose that there exists $\uv_1^*,\uv^*_2\in S^*$ distinct such that $\psi(\uv_1^*)=\psi(\uv_2^*)$, so $\F_p(\yv^*-\uv_1^*) = \F_p(\yv^*-\uv_2^*)$. This means that there exists a nonzero vector $\kv\in\Z^n$ such that $\yv^*-\uv_1^* = \yv^*-\uv_2^* + p\kv$. So by triangular inequality and \autoref{lem:closestNorm} we have
    \begin{align*}
        \enorm{\yv^*-\uv_1^*}\quad &\geq\quad \enorm{\yv^*-\uv_2^*} + p\enorm{\kv}\\
        \sqrt{n}\quad\quad &\geq\quad\quad \sqrt{n}\quad\quad+ p\enorm{\kv}
    \end{align*}
    This means that $\kv = \Omat$. So we obtain a contradiction proving the injectivity.
    Conversely, we obviously have $\psi^{-1}:\widetilde{\ev}\longrightarrow\yv^*-\Z\left(\widetilde{\ev}\right)$ which is injective.

    Because $S^*$ and $\widetilde{S}$ are finite sets, we now just need to show that $\psi(S^*)\subset\widetilde{S}$ and that $\psi^{-1}\left(\widetilde{S}\right)\subset S^*$. The first point is immediate given \autoref{lem:EuclToHamm} and \autoref{lem:closestNorm}. For the second one, let $\widetilde{\ev}\in\widetilde{S}$ and $\ev^* \eqdef \Z(\widetilde{\ev})$. By \autoref{lem:HammToEucl} we know that $\uv^* \eqdef \yv^*-\ev^* = \psi^{-1}\left(\widetilde{\ev}\right)$ belongs to $\LL^*$ and that $\enorm{\ev^*} = \sqrt{n}$. So by \autoref{lem:closestNorm} we have $\uv^*\in S^*$.\qed
\end{proof}
We now have everything necessary to build the reduction from \ResSD{} to~\CVP{}.

\begin{theorem}
\label{st:ressd-to-cvp}
    Let $(\Hm,\sv)\in\F_p^{(n-k)\times n}\times\F_p^{n-k}$ be an instance of \ResSD{} with solution entries in $E = \pset{r_1,\dots r_z}$, and consider $\left(\LL\left(\widetilde{\CC}\right),\yv^*\right)$ where $\widetilde{\CC}\eqdef\CC^\perp\left(\widetilde{\Hm}\right)$ and $\F_p(\yv^*)\transp{\widetilde{\Hm}} = \widetilde{\sv}$, as an instance of \CVP{}. If there exists an algorithm which can solve $\left(\LL\left(\widetilde{\CC}\right),\yv^*\right)$ in time $T$ with success probability $\rho$, then there exists an algorithm solving $(\Hm,\sv)$ in time $T + \mathcal{O}(n^3)$ with success probability $\rho$.
\end{theorem}

\begin{proof}
    Computing $\left(\widetilde{\Hm},\widetilde{\sv}\right)$ can be done in quadratic time, and computing $\yv^*$ needs to solve a linear system, which can be done in time $\mathcal{O}(n^3)$. Let us then consider $\varphi$ and $\psi$ the mappings as defined in \autoref{lem:ResRegBij} and \autoref{prop:solBij} respectively. We then have $\varphi\circ\psi^{-1}$ which is a bijection between the solutions of both of our problems, and computing a preimages of this mapping can be done in linear time.\qed
\end{proof}

\begin{remark}
    This reduction comes from the same ideas as those of \cite{EPRINT:HKNSC24}, where reductions are given from Lee-based problems to lattice-based problems. But in our case, the distributions of lattices are particular (not uniform).
\end{remark}


\subsection{Security analysis - Hybrid-BatchCVP attack}
\label{sec:security-batchcvp}

Combining the reduction with modern heuristic algorithms for Batch-CVP yields the following generic result.

\begin{corollary}
Let $(\Hm,\sv)\in\F_p^{(n-k)\times n}\times\F_p^{n-k}$ be an instance of \ResSD{} with solution entries in $E$ with $|E|=z$. Then, under \autoref{heur:cvp-cost}, it can be solved in time $2^{0.234z(n-g)+o(zn)}z^g$ and memory $2^{0.208z(n-g)+o(zn)}$, where $g = \lceil 
0.058zn\log_z{2}\rceil$.
\end{corollary}
\begin{proof}
Let us guess $g$ first blocks of $z$ variables each in the final lattice ($z^g$ candidates) and modify the target vector each time accordingly. We obtain a batch of $M=z^g$ CVP target vectors for a single lattice of reduced dimension $z(n-g)$ (the guessing can be equivalently done on the initial \ResSD{} problem, so that the reduction proof remains valid for the smaller problem). Note that $M \ge 2^{0.058zn} \ge 2^{0.058z(n-g)}$ and therefore \autoref{heur:cvp-cost} is applicable.
\qed
\end{proof}
\begin{remark}
Since the number of guesses is chosen such that $2^{0.234zn}z^g \approx 2^{0.292zn}$, the reduction from the use of the Batch-\CVP{} technique (compared to basic sieving cost for \CVP{}) comes from the reduction of the lattice dimension, and is equal to $2^{0.234zg}\approx 2^{0.0014z^2 n \log_z{2}}$.
\end{remark}
\begin{table}[htb!]
\centering
\setlength{\tabcolsep}{0.5em}
\begin{tabular}{cc|p{1.2cm}p{1.5cm}p{1.7cm}|p{1.2cm}l}
\toprule
$(n,k)$ & $z'$ & Trunc. prob. & $g$ blocks guessed & Batch-\CVP{} cost & Total time & Memory \\
\toprule
\multirow{7}{*}{(127, 76)} 
& 7 & $2^{-0.0}$ & 19 ($2^{53.3}$)& $2^{230.2}$ & $2^{231.5}$& $2^{157.2}$ \\
& 6 & $2^{-28.2}$ & 18 ($2^{46.5}$)& $2^{199.6}$ & $2^{229.1}$& $2^{136.0}$ \\
& 5 & $2^{-61.6}$ & 16 ($2^{37.2}$)& $2^{167.0}$ & $2^{229.9}$& $2^{115.4}$ \\
& 4 & $2^{-102.5}$ & 15 ($2^{30.0}$)& $2^{134.8}$ & $2^{238.6}$& $2^{93.2}$ \\
& 3 & $2^{-155.2}$ & 14 ($2^{22.2}$)& $2^{101.5}$ & $2^{258.0}$& $2^{70.5}$ \\
& 2 & $2^{-229.5}$ & 15 ($2^{15.0}$)& $2^{67.4}$ & $2^{298.2}$& $2^{46.6}$ \\
\midrule
\multirow{7}{*}{(187, 111)} 
& 7 & $2^{-0.0}$ & 28 ($2^{78.6}$)& $2^{339.0}$ & $2^{340.3}$& $2^{231.5}$ \\
& 6 & $2^{-41.6}$ & 26 ($2^{67.2}$)& $2^{293.3}$ & $2^{336.1}$& $2^{200.9}$ \\
& 5 & $2^{-90.8}$ & 24 ($2^{55.7}$)& $2^{246.4}$ & $2^{338.5}$& $2^{169.5}$ \\
& 4 & $2^{-151.0}$ & 22 ($2^{44.0}$)& $2^{198.4}$ & $2^{350.7}$& $2^{137.3}$ \\
& 3 & $2^{-228.6}$ & 21 ($2^{33.3}$)& $2^{149.8}$ & $2^{379.7}$& $2^{103.6}$ \\
& 2 & $2^{-338.0}$ & 22 ($2^{22.0}$)& $2^{99.2}$ & $2^{438.5}$& $2^{68.6}$ \\
\midrule
\multirow{7}{*}{(251, 150)} 
& 7 & $2^{-0.0}$ & 37 ($2^{103.9}$)& $2^{454.4}$ & $2^{455.7}$& $2^{311.6}$ \\
& 6 & $2^{-55.8}$ & 34 ($2^{87.9}$)& $2^{392.6}$ & $2^{449.6}$& $2^{270.8}$ \\
& 5 & $2^{-121.8}$ & 32 ($2^{74.3}$)& $2^{330.5}$ & $2^{453.6}$& $2^{227.8}$ \\
& 4 & $2^{-202.6}$ & 30 ($2^{60.0}$)& $2^{266.9}$ & $2^{470.8}$& $2^{183.9}$ \\
& 3 & $2^{-306.8}$ & 28 ($2^{44.4}$)& $2^{200.9}$ & $2^{509.0}$& $2^{139.2}$ \\
& 2 & $2^{-453.6}$ & 30 ($2^{30.0}$)& $2^{133.4}$ & $2^{588.3}$& $2^{91.9}$ \\
%
%
%
%
\bottomrule
\end{tabular}
\vspace{0.25em}
\caption{Hybrid-Batch-\CVP{} cost estimation for CROSS-R-SDP. The attack consists in probabilistic truncation from the initial $z=7$ to $z'$, guessing $g$ coordinate blocks, and running heuristic Batch-\CVP. The complexities are adapted to achieve success rate of 90\%. Up to the correction factor, the total time complexity is the product of the inverse of the truncation probability and the batch-\CVP{} cost.}
\label{tab:cvp-cost}
\end{table}
We consider the application of the hybrid attack - multiplicative truncation (see \autoref{st:mult-truncation}), \CVP{} reduction, partial guessing, and Batch-\CVP{} solution - to CROSS-R-SDP parameters and their reduced variants. The results are summarized in \autoref{tab:cvp-cost}. Interestingly, the best time complexity is achieved at truncation $z'=6$ (out of the original $z=7$); it is roughly equal to $2^{1.8n}$. Choosing $z'=5$ leads to a small slowdown, but significant savings in memory complexity.

However, we note that,  at least for the CROSS-R-SDP parameters,
the proposed attack complexities do not improve over basic memoryless ISD-like guessing of $k$ coordinates in the initial problem \ResSD{}. This attack has time complexities $2^{213.4}, 2^{311.6}, 2^{421.1}$ for $n=127,187$, and $251$ respectively. The main contribution to this fact is that we used generic CVP solvers whose complexity estimate depends only on the dimension of the lattice, which is quite large in our case ($zn$). Possible analysis based on the length of the solution vector is made difficult due to the existence of a large number of short vectors in the lattice. For example, for the restriction used in CROSS, any vector of the form $(\ldots,0,2,-3,1,0,\ldots)$ belongs to the lattice, where the part $(2,-3,1)$ is inside any block. Developing and analyzing tailored lattice algorithms is an interesting avenue for future work.

\TODO{Some small scale experiments. Maybe not priority, since even for $n=50$ and small $z'$ the cost is high. (large dimension). Still confirming that it all makes sense would be good.}


\section{Lattice-based Modelings: Reduction of ResSD to List-CVP and List-SVP without expansion}
\label{sec:list-cvp}

We also propose another heuristic, but more direct and more compact lattice-based modeling of \ResSD{}. The key idea is that the set $E$ of allowed values can be modeled by a set of ``small'' values around a well-chosen center. Furthermore, the values can be scaled by a nonzero constant which increases the number of such embeddings. The effect is stronger when the set is small or has a certain structure. For this purpose, we will 
use the multiplicative truncation (\autoref{st:mult-truncation}) to reduce the size of the set~$E$.

We start with the notion of \emph{an affine diameter} in \autoref{sec:affine-diameter}, which quantifies how close the values of the set $E$ can be packed together under an affine substitution. In \autoref{sec:listcvp-general-reduction}, we describe the general shape of the reduction with an arbitrary center $\mu$ and a radius $R$. Then, in \autoref{sec:listcvp-aff-diam}, we analyze the specialization of the reduction based on the affine diameter, which leads to a concrete deterministic reduction to \ListCVP. In \autoref{sec:listcvp-mean-median}, we continue with a stronger but probabilistic reduction to \ListCVP{} based on the standard deviation of the set $E$ and the mean of $E$ as the center. We show in \autoref{sec:listcvp-svp} that for an integral center the problem can be converted to \ListSVP{} by manipulating the \ResSD{} instance before the reduction. Remarkably, for some \ResSD{} instances (low code rate and compact restriction set),
one can heuristically deduce reductions to pure \CVP{} and \SVP{}. Finally, in \autoref{sec:security-listcvp} we conclude with applying the reductions to CROSS-R-SDP instances and verifying them experimentally.


\subsection{Affine diameter}
\label{sec:affine-diameter}

\begin{definition}
Let $E \subseteq \F_p$, $|E|=z>0$. We define the \emph{affine diameter} $D_E$ of the set as 
\begin{equation}
\label{eq:ad-minmax}
\min_{a,b\in\F_p, a\ne 0}
~\max_{x \in E}
~\Z(ax+b).
\end{equation}
\end{definition}

\begin{remark}
For computational purposes, a small optimization is to use the equivalent expression \[
D_E = \min_{a\in\F_p\setminus\{0\}, x_0 \in E}
~\max_{x \in E}
~\Z(ax-ax_0),
\]
which reduces the number of $b$ to be tested from $p$ to $|E|$. To see that it is true, observe that, for a fixed $a$, the minimum in \eqref{eq:ad-minmax} is achieved when one of the elements of $E$ is mapped to 0. Otherwise, one can subtract the smallest integer from $\Z(ax+b)$ without wrapping modulo $p$, decreasing the maximum value by that amount. Therefore, it is sufficient to enumerate all nonzero $a$ and an element $x_0 \in E$ to be mapped to 0, and choose the minimum value of the sets' maximums.
\end{remark}
It is easy to see that $D_E$ is the smallest length of a continuous segment that contains all elements of $aE+b$. We will show how to use a small affine diameter to convert \ResSD{} into a lattice-based problem. Before that, we will study the notion itself further.

\medskip
The case of $z=1$ is trivial with $D_E=0$. 
Generally, any two elements of $E$ may always be mapped to 0 and 1 by a valid choice of $a,b$ in the definition. In the case of $z=2$, this implies $D_E=1$, which essentially leads to a subset-sum formulation of \ResSD, as already noted by the designers of CROSS \cite{NISTPQC-ADD-R2:CROSS24}.

The larger cases of $z > 2$ become nontrivial. We study some examples computationally for $p=127$ used in CROSS.
Note that the set $E$ used in CROSS has subsets of size $z'$ with $1 \le z' \le 7$, of affine diameter $2^{z'-1}-1$, namely $\{1,2,4,\ldots,2^{z'-1}\} \subseteq E$. We will use this together with the truncation technique. The results are summarized in \autoref{tab:cross-ad}.

For example, in the case of $z=3$ with $p=127$ as in CROSS, without loss of generality, assume the shape $E=\{0,1,a\}$. The average affine diameter (over $a \in \pset{2,\ldots,p-1} \subseteq \F_p$) is $8.032$. The highest value of $D_E$ is 13 achieved with $\pset{0,1,20}$, which is affine-equivalent to $\pset{0,6,13}$. The other possible values of $D_E$ are all integers from 2 to 11 and are achieved (for example) by the sets $\pset{0,1,D_E)}$.

\begin{table}[htb]
\setlength{\tabcolsep}{0.5em}
\centering
\begin{tabular}{c|ccccc}
\toprule
 & \multicolumn{5}{c}{Affine diameter $D_E$} \\
$z=|E|$ & CROSS & Min. & Avg. & Max. & Example (Max.) \\
\midrule
2 & 1 & 1 & 1 & 1 & $\pset{0,1}$\\
3 & 3 & 2 & 8.03 & 13 & $\pset{0,6,13}$\\
4 & 7 & 3 & 17.81 & 27 & $\pset{0,11,12,27}$\\
5 & 15 & 4 & 27.47 & 42 & $\pset{0, 12, 14, 17, 42}$\\
6 & 31 & 5 & 36.09 & 57 & $\pset{0, 12, 13, 16, 23, 57}$ \\
7 & 63 & 6 & 43.56 & 65 & $\pset{0, 22, 35, 40, 54, 63, 65}$ \\
8 & - & 7 & $50.01^{\dagger}$ & $72^{\dagger}$ & $\pset{0, 16, 27, 31, 32, 35, 38, 72}$ \\
9 & - & 8 & $55.60^{\dagger}$ & $79^{\dagger}$ & $\pset{0, 1, 2, 4, 7, 13, 22, 46, 79}$\\
10 & - & 9 & $60.42^{\dagger}$ & $83^{\dagger}$ & $\pset{0, 18, 28, 29, 36, 56, 59, 71, 82, 83}$\\
\bottomrule
\end{tabular}
\vspace{0.25em}
\caption{Comparison of the affine diameter of subsets of $E$ used in CROSS-R-SDP ($p=127$) to average/maximum values over subsets of $\F_p$ of size $z$.
\\
$\dagger$: approximate values obtained by sampling at least $10^7$ random subsets of size $z$.
}
\label{tab:cross-ad}
\end{table}

We conclude that the affine diameter is a nontrivial parameter with a broad range of values. We observed that the values for the subsets of the set $E$ used in CROSS are on the smaller side, showing suboptimal resistance. Together with the truncation technique, we will use it to derive new direct lattice-based models of the \ResSD{} and CROSS in particular.

\subsection{Reduction to List-CVP}
\label{sec:listcvp-general-reduction}
The following theorem shows the close relation between solution vectors to a (restricted) syndrome decoding problem and the lattice associated to the corresponding linear code (as defined in \autoref{sec:cvp}). One has to keep \autoref{st:affine-change} (Affine Shifting) in mind, since it allows us to change the set $E$ which can lead to different behavior of the Euclidean distances.

\begin{theorem}\label{st:ressd-to-lattice}
Let $(\Hm,\sv)\in\F_p^{(n-k)\times n}\times\F_p^{n-k}$ be an instance of \ResSD{}
with solution entries in a set $E$.
Let positive $\mu \in \R$. Then, there exists an $n$-dimensional lattice $\LL \subseteq \Z^n$ of volume $p^{n-k}$ and a vector $\vv \in \R^n$, constructible in time polynomial in $n$, such that there exists an embedding $\varphi$ of the solutions $\ev \in \F_p^n$, $\ev \Hm^T=\sv$, into the lattice, such that
\begin{equation}
\label{eq:code-lattice-distance}
\enorm{\varphi(\ev) - \vv} =
\sqrt{
\sum_{i=1}^n \pround{\Z(e_i) - \mu}^2
}.
\end{equation}
This embedding and its left inverse are computable in time $\mathcal{O}(n)$.
\end{theorem}
\begin{proof}
We first construct $\LL, \yv$ and $\varphi$.
Let $\CC$ be the linear code that admits $\Hm$ as a parity check matrix.
Let $\av \in \F_p^n$ be an arbitrary solution to $\av \Hm^T = \sv$.
Define 
\begin{align}
\LL &= \LL(\CC) ~~ \subseteq \Z^n,\\
\vv &= \vec{\mu} - \Z(\av) ~~ \in \R^n,\\
\varphi: \ev &\mapsto \Z(\ev) - \Z(\av) ~~ \in \Z^n,
\end{align}
where $\vec{\mu}=(\mu,\ldots,\mu) \in \R^n$.

We now prove the claims. 
Clearly, $\LL$ and $\vv$ are constructible efficiently.
We now prove that $\varphi$ produces a vector in $\LL$.
Let $\ev \in \F_p^n$ be a solution to the SD problem: $\ev \Hm^T = \sv$.
Then, \[
\F_p(\varphi(\ev))\Hm^T = (\ev - \av)\Hm^T = 0,
\]
and so $\F_p(\varphi(\ev))$ belongs to $\CC$, which by \autoref{prop:lcbij} implies that $\varphi(\ev)$ belongs to $\LL$.
Furthermore, we have \[
\varphi(\ev) - \vv = \Z(\ev) - \Z(\av) - \vec{\mu} + \Z(\av) = \Z(\ev) - \vec{\mu}
\]
which implies \eqref{eq:code-lattice-distance}.
Finally, the map $\psi: \yv \mapsto \F_p(\yv + \av)$ is a left inverse of $\varphi$ and it is clear that both $\phi$ and $\psi$ can be computed in linear time.\qed
\end{proof}
As a corollary, we derive the generic reduction of \ResSD{} to \ListCVP, which is parametrized by the ``restriction center'' $\mu$ which should be close to all elements of $E$, and the radius $R$ bounding the distance from this center. 

\begin{corollary}\label{st:list-cvp-reduction}
Let $(\Hm,\sv)\in\F_p^{(n-k)\times n}\times\F_p^{n-k}$ be an instance of \ResSD{} with solution entries in a set $E$.
Let positive $\mu,R \in \R, \mu \ge 0, R > 0$, and define
\begin{equation}
\label{eq:listcvp-success}
\rho = \Prob_{\substack{e_i \in \Z(E)\\1 \le i \le n}}\pset{\sum_{i=1}^n (e_i-\mu)^2 \le R^2}.
\end{equation}
Let a lattice $\LL \subseteq \Z^n$ and a vector $\vv \in \R^n$ be defined as in \autoref{st:ressd-to-lattice}.
Then, with probability $\rho$ over solutions to the \ResSD{} problem,
a solution to the \ListCVP{} instance $(\LL,\vv)$ with distance upper bound $R$ and $N$ output vectors
can be converted into a solution to the \ResSD{} instance in time $\mathcal{O}(nN)$. 
\end{corollary}
\begin{proof}
From \autoref{st:ressd-to-lattice} it is clear that solutions $\ev$ satisfying the bound in \eqref{eq:listcvp-success} belong to the constructed \ListCVP{} instance within radius $R$. Recovering the \ResSD{} solution can be done in linear time per candidate vector, since it only requires addition of a constant vector and testing the restriction to $E$.
\qed
\end{proof}

\if10  
\begin{remark}
In fact, a complete solution of \ListCVP{} is not necessary. A solution that misses a random fraction $\epsilon$ of the vectors would simply decrease the probability of success by a factor $(1-\epsilon)$.
\end{remark}
\fi

\begin{remark}
\label{rem:pure-cvp}
Under the Gaussian Heuristic and using \autoref{prop:vol}, we note that when $R \le \sqrt{\frac{n}{2\pi e}} \Vol^{1/n}= \sqrt{\frac{n}{2\pi e}} p^{1-k/n}$, the  number of vectors in the lattice satisfying the distance bound $R$ is a small constant. In this case, one can heuristically reclassify the \ListCVP{} formulation as a simple \CVP{}. In contrast, \autoref{st:ressd-to-cvp} is a non-heuristic precise reduction to \CVP.
\end{remark}

\begin{remark}
Note that the success probability is taken over the instances of the initial problem, and so is essentially a weak-key attack when $\rho < 1$.
For a fixed $\mu$, by varying the distance upper bound $R$ we can explore different trade-off points between the \ListCVP{} complexity and the success probability. However, reaching a very high success probability would be extremely costly. Luckily, we have an alternative in the case where $E$ is a multiplicative subgroup, as in CROSS. 
The natural randomization method (\autoref{st:mult-randomization}) allows the reduction to be performed multiple times, each time having a fresh success probability $\rho$. Here, one needs to note that the reduction probability in the theorem is over the possible error-vector solutions, and the randomization technique perfectly randomizes these vectors.
\end{remark}


\subsection{Analysis based on affine diameter}
\label{sec:listcvp-aff-diam}
A possible choice for $\mu$ and $R$ is naturally based on the affine diameter. The following proposition describes a precise compact reduction based on this quantity.

\begin{proposition}
Let $(\Hm,\sv)\in\F_p^{(n-k)\times n}\times\F_p^{n-k}$ be an instance of \ResSD{} with solution entries in a set $E$. Then, in time polynomial in $n$ and linear in the output size, this problem can be reduced to a \ListCVP{} problem over an $n$-dimensional lattice of volume $p^{n-k}$ and distance $\sqrt{n}D_E/2$.
\end{proposition}
\begin{proof}
We first apply \autoref{st:affine-change} to obtain an equivalent \ResSD{} instance with $E'=\pset{0,\ldots,D_E}$ and then apply \autoref{st:list-cvp-reduction} with $\mu=D_E/2$ and $R=\sqrt{n}D_E/2$, so that \[
\sum_{i=1}^n (e_i - \mu)^2 \le 
\sum_{i=1}^n \pround{\frac{D_E}{2}}^2 \le \frac{nD_E^2}{4},
\]
implies the success probability $\rho=1$ exactly. The lattice from \autoref{st:list-cvp-reduction} (and \autoref{st:ressd-to-lattice}) has volume $p^{n-k}$ by \autoref{prop:vol}.
\qed
\end{proof}
An interesting question one may ask is for which \ResSD{} parameters $n,k,p,D_E$ the problem reduces to pure \CVP{} (which would be solvable in heuristic time $2^{0.292n+o(n)}$ by \autoref{heur:cvp-cost})? The following analysis is under the GH.

The expected length of the shortest nonzero vector (or distance from a random point) is $\sqrt{\frac{n}{2\pi e}}\Vol^{1/n}=\sqrt{\frac{n}{2\pi e}} p^{\frac{n-k}{n}}$. When this quantity is larger than $R=\sqrt{n} D_E/2$, we can expect a unique solution coinciding with the image of the solution to the \ResSD{} problem under the reduction. Therefore, we need \[
D_E < \sqrt{\frac{2}{\pi e}} p^{1-k/n} \approx 0.48 p^{1-k/n}.
\]
This bound benefits from smaller $D_E$, larger $p$, and lower code rate.

Regarding the overall sparsity, a natural requirement is to expect roughly 1 random solution to the problem, which means $z^n \approx p^{n-k}$, $z=|E|$ (see \cite{NISTPQC-ADD-R2:CROSS24}), and thus the bound becomes $
D_E < 0.48 z$.
Clearly, this is not possible for $z \ge 3$. Therefore, this reduction with $(\mu,R)$ based on the affine diameter never reduces to pure \CVP{} in interesting cases. However, the reduction may still be useful for two reasons. First, we can apply multiplicative truncation (\autoref{st:mult-truncation}), which makes the problem sparser by eliminating some elements of $E$. The combination of truncation and reduction may thus be more powerful. Second, the \ListCVP{} formulation can still lead to efficient attacks, since the enumeration complexity gradually increases from sieving $2^{0.292n+o(n)}$, with increasing \ListCVP{} radius. In the next subsection, we will evaluate these options, but together with better choices of $\mu$ and $R$.


\subsection{Analysis based on the standard deviation of E}
\label{sec:listcvp-mean-median}

Although setting $\mu,R$ in the reduction based on the affine diameter $D_E$ is a natural option and leads to good theoretical bounds and deterministic success ($\rho=1$), stronger attacks may be achieved when $\rho < 1$ is allowed. 

Another natural choice of $\mu$ is the one that minimizes the expected value of the distance (where the expectation is over $\ev \in E^n$) \[
\enorm{\varphi(\ev) - \vv} =
\sqrt{\sum_{i=1}^n \pround{\Z(e_i) - \mu}^2}.
\]
Since it is not straightforward to compute, 
we will first minimize the expected square of the distance instead.

\begin{proposition}
\label{st:mean-squared}
The expected value of the sum $\sum_{i=1}^n \pround{\Z(e_i) - \mu}^2$, where each $e_i$ is sampled independently and uniformly from a set $E$, is minimized at $\mu=\frac{1}{|E|}\sum_{r \in E} \Z(r)$.
\end{proposition}
\begin{proof}
By linearity of expectation, \[
\mathbb{E}\pset{
    \sum_{i=1}^n \pround{\Z(e_i) - \mu}^2
}
=
n
\mathbb{E}\pset{
     \pround{\Z(e_0) - \mu}^2
}
= \frac{n}{|E|}\sum_{r\in E}\pround{\Z(r)-\mu}^2,
\]
where the second equation is due to $e_0$ being sampled uniformly from $E$.
The latter expression is minimized at $\mu=\frac{1}{|E|}\sum_{r \in E} \Z(r)$.
\qed
\end{proof}
Setting $\mu$ to the average value of $\Z(E)$ means that the expected value of the sum $\sum_{i=1}^n \pround{\Z(e_i) - \mu}^2$ being minimized is equal to $n\sigma^2(E)$ where $\sigma^2(\Z(E))$ is the variance of a random variable sampled uniformly from $E$.

For computing the success probability of the reduction, it is more convenient to work with the median of the distance instead of the mean. This would give a good estimate on the \ListCVP{} distance in the reduction with the success rate $\rho=50\%$. For this purpose, we will heuristically assume that the mean approximates well the median. Note that taking the median commutes with squaring, so that the median of the distance is equal to the square root of the median of the squared distance, which by the heuristic coincides with the mean of the squared distance, which is in turn minimized and calculated in \autoref{st:mean-squared}. We will also evaluate the quality of this heuristic experimentally.

\begin{heuristic}\label{heur:median}
The median value $M_{E,n}$ of the distance $d = \enorm{\varphi(\ev) - \vv}$ (using $\mu=\frac{1}{|E|}\sum_{r \in E} \Z(r)$) is well approximated by the square root of the mean of $d^2$: \[
M_{E,n} = \sqrt{\frac{n}{|E|} \cdot
\sum_{r \in E}^z {(\Z(r) - \mu)^2}}.
\]
\end{heuristic}

\medskip

\if10
We now show that the full distribution of the variable $X$ can be efficiently computed for a wide range of parameters.

\begin{proposition}
Let $X_i, 1 \le i \le n$, be independent variables uniformly distributed from $E \subseteq \Z$, $0 = \min{E} \le \max{E} = D_E$, $\mu = \frac{1}{|E|} \sum_{r_i\in E} r_i$.
Then, the distribution of the variable $X=\sqrt{\sum_{i=1}^n (X_i - \mu)^2}$ can be computed in time $\tilde{\mathcal{O}}(|E|^2 D_E^2 n)$.
\end{proposition}
\begin{proof}
It is sufficient to compute the distribution of $X^2 = \sum_{i=1}^n (X_i - \mu)^2$. First, let $z=|E|$ as usual and note that \[
z^2X^2 = \sum_{i=1}^n (zX_i - z\mu)^2
\]
is an integer variable, lower bounded by 0 and upper bounded by $U\eqdef nz^2D_E^2$. Define $S = \{(zr_i - z\mu)^2 \mid r_i \in E\}$ so that $z^2X^2$ is simply a sum of $n$ independent variables $Y_i$ sampled uniformly from $S$. The generating function of each $Y_i$ is then \[
f(x) = \frac{1}{z}\sum_{d \in F} x^d,
\]
and that of $z^2X^2$ is simply $f(x)^n$. Here, the coefficient of $x^j$ in $f(x)$ (resp. $f(x)^n$) is the probability that $Y_i = j$ (resp. $z^2X^2 = j$).
This polynomial $f(x)^n$ has degree at most $U$, and can be efficiently computed in time $\mathcal{O}(U\log{U})$ using Fast Fourier Transform (FFT): one forward transformation, one coordinate-wise exponentiation to the power $n$, and one inverse transformation.
\end{proof}

\fi


\subsection{Reduction to List-SVP}
\label{sec:listcvp-svp}

In this subsection, we describe a trick that allows one to convert the problem into \ListSVP{} / \SVP{} using the fact that the problem originates from linear codes, where the corresponding decoding and low-weigh codeword problems are more closely related.

\begin{proposition}
\label{st:ressd-to-lattice-svp}
Let $(\Hm,\sv)\in\F_p^{(n-k)\times n}\times\F_p^{n-k}$ be an instance of \ResSD{}
with solution entries in a set $E$. Let $\mu$ be a positive \emph{integer}. Then, there exists an $n$-dimensional lattice $\LL \subseteq \Z^n$ of volume $p^{n-k-1}$, constructible in time polynomial in $n$, such that there exists an embedding $\varphi$ of the solutions $\ev \in \F_p^n$, $\ev \Hm^T=\sv$, into the lattice, such that
\begin{equation}
\label{eq:code-lattice-norm}
\enorm{\varphi(\ev)} =
\sqrt{
\sum_{i=1}^n \pround{\Z(e_i) - \mu}^2
}.
\end{equation}
This embedding and its left inverse are computable in time $\mathcal{O}(n)$.
\end{proposition}
\begin{proof}
First, we apply \autoref{st:affine-change} with $a=1$ and $b=-\mu$ to obtain a new \ResSD{} problem $(\Hm',\sv')$ with restriction $E'=\{r-\mu \mid r \in E\}$.
Second, choose an arbitrary invertible matrix $\Qm \in \F_p^{(n-k)\times(n-k)}$ such that
$$\Qm \sv^{\prime T} = (1,0,\ldots,0)^T = \sv^{\prime\prime T}.$$
Let $\Hm'' = \Qm \Hm'$. Clearly, the sets of solutions of $\ev\Hm^{\prime T}=\sv'$ and of $\ev\Hm^{\prime \prime T}=\sv''$ are identical. Now set $\Hm''' \in \F_p^{(n-k-1)\times n}$ to $\Hm''$ with the first row removed, which only increases the set of solutions in the new \ResSD{} instance $\ev\Hm'''=0$ (by an expected factor of $p$).\\
\\
Let $\CC$ be the linear code that admits $\Hm'''$ as a parity check matrix.
Define 
\begin{align}
\LL &= \LL(\CC) ~~ \subseteq \Z^n,\\
\varphi: \ev &\mapsto \Z(\ev) - \vec{\mu} ~~ \in \Z^n,
\end{align}
where $\vec{\mu}=(\mu,\ldots,\mu) \in \R^n$.
Clearly, $\LL$ is constructible efficiently.
Let $\ev \in \F_p^n$ be a solution to the original SD problem: $\ev \Hm^T = \sv$.
Then,
\begin{align*}
\Hm'' \F_p(\varphi(\ev))^T =
\Qm\Hm' \F_p(\varphi(\ev))^T = \Qm\Hm'(\ev - \vec{\mu}) = \Qm\sv^{\prime T} = (1,0,\ldots,0)^T,
\end{align*}
and so $\F_p(\varphi(\ev))$ belongs to $\CC$, which by \autoref{prop:lcbij} implies that $\varphi(\ev)$ belongs to $\LL$. The norm equation \eqref{eq:code-lattice-norm} follows directly from the definition of $\varphi$. It is evident that $\phi$ has left inverse and both are computable in linear time.
\qed
\end{proof}
This directly implies the \ListSVP{} reduction, with proof as in \autoref{st:list-cvp-reduction}.

\begin{corollary}\label{st:list-svp-reduction}
Let $(\Hm,\sv)\in\F_p^{(n-k)\times n}\times\F_p^{n-k}$ be an instance of \ResSD{} with the solution's entries in a set $E$.
Let a nonnegative integer $\mu$ and a positive real $R$ and define
\[
\rho = \Prob_{\substack{e_i \in \Z(E)\\1 \le i \le n}}\pset{\sum_{i=1}^n (e_i-\mu)^2 \le R^2}.
\]
Let a lattice $\LL \subseteq \Z^n$ be defined as in \autoref{st:ressd-to-lattice-svp}.
Then, with probability $\rho$ over solutions to the \ResSD{} problem,
a solution to \ListSVP{} over $\LL$ with distance upper bound $R$ and $N$ output vectors
can be converted into a solution to the \RegSD{} instance in time $\mathcal{O}(N)$.
\end{corollary}
This reduction is of theoretical and practical interest.
For parameters implying $R < \lambda_1(\LL)$, under GH, the constructed \ListSVP{} problem degenerates to a pure \SVP{} problem, which is an interesting result for a code-based problem.
The reduction provides an alternative to the Kannan embedding \cite{STOC:Kannan83a}, which is the most used technique for converting a \CVP{} instance into an \SVP{} instance. \SVP{} algorithms are often more performant than \CVP{} in practice. However, in our analysis below, we only consider strongest generic methods (\autoref{st:sieving} and \autoref{heur:cvp-cost}) which have the same complexity for \SVP{} and \CVP{} (perhaps having a small difference hidden in the omitted polynomial factors).


\subsection{Security analysis - Hybrid-ListCVP attack}
\label{sec:security-listcvp}
We discuss the hybrid attack - multiplicative truncation and \ListCVP{} reduction - applied to CROSS-R-SDP parameters. We are using the mean-based $\mu$ and the median $R$ from \autoref{sec:listcvp-mean-median}.
The results are summarized in \autoref{tab:list-cvp-cost}.

\begin{table}[htb!]
\centering
\setlength{\tabcolsep}{0.5em}
\begin{tabular}{cc|p{1.1cm}lp{1cm}p{1cm}|p{1cm}l}
\toprule
$(n,k)$ & $z'$ & Trunc. prob. & $|\LL \cap \Ball{n,R}|$ & Sieve cost & Enum. cost & Total time & Memory \\
\toprule
\multirow{7}{*}{(35, 21)} 
& 7 & $2^{-0.0}$ & $2^{124.7}$ & $2^{10.2}$ & $2^{130.0}$ & $2^{136.0}$& $2^{7.3}$ \\
& 6 & $2^{-7.8}$ & $2^{90.7}$ & $2^{10.2}$ & $2^{96.2}$ & $2^{110.0}$& $2^{7.3}$ \\
& 5 & $2^{-17.0}$ & $2^{56.1}$ & $2^{10.2}$ & $2^{62.0}$ & $2^{85.0}$& $2^{7.3}$ \\
& 4 & $2^{-28.3}$ & $2^{20.2}$ & $2^{10.2}$ & $2^{28.0}$ & $2^{62.2}$& $2^{7.3}$ \\
& 3 & $2^{-42.8}$ & $2^{-18.4}$ & $2^{10.2}$ & $2^{6.8}$ & $2^{57.0}$& $2^{7.3}$ \\
& 2 & $2^{-63.3}$ & $2^{-64.6}$ & $2^{10.2}$ & $2^{3.6}$ & $2^{77.5}$& $2^{7.3}$ \\
& 1 & $2^{-98.3}$ & $2^{-\infty}$ & $2^{10.2}$ & $2^{0.0}$ & $2^{112.5}$& $2^{7.3}$ \\
\midrule
\multirow{7}{*}{(127, 76)} 
& 7 & $2^{-0.0}$ & $2^{459.1}$ & $2^{37.1}$ & $2^{466.5}$ & $2^{472.5}$& $2^{26.4}$ \\
& 6 & $2^{-28.2}$ & $2^{335.7}$ & $2^{37.1}$ & $2^{343.5}$ & $2^{377.7}$& $2^{26.4}$ \\
& 5 & $2^{-61.6}$ & $2^{210.1}$ & $2^{37.1}$ & $2^{219.4}$ & $2^{287.1}$& $2^{26.4}$ \\
& 4 & $2^{-102.5}$ & $2^{79.9}$ & $2^{37.1}$ & $2^{105.2}$ & $2^{213.8}$& $2^{26.4}$ \\
& 3 & $2^{-155.2}$ & $2^{-60.3}$ & $2^{37.1}$ & $2^{30.3}$ & $2^{196.3}$& $2^{26.4}$ \\
& 2 & $2^{-229.5}$ & $2^{-227.8}$ & $2^{37.1}$ & $2^{7.1}$ & $2^{270.6}$& $2^{26.4}$ \\
& 1 & $2^{-356.5}$ & $2^{-\infty}$ & $2^{37.1}$ & $2^{0.0}$ & $2^{397.6}$& $2^{26.4}$ \\
\midrule
\multirow{7}{*}{(187, 111)} 
& 7 & $2^{-0.0}$ & $2^{671.4}$ & $2^{54.6}$ & $2^{679.4}$ & $2^{685.4}$& $2^{38.9}$ \\
& 6 & $2^{-41.6}$ & $2^{489.7}$ & $2^{54.6}$ & $2^{498.4}$ & $2^{545.9}$& $2^{38.9}$ \\
& 5 & $2^{-90.8}$ & $2^{304.8}$ & $2^{54.6}$ & $2^{316.5}$ & $2^{413.3}$& $2^{38.9}$ \\
& 4 & $2^{-151.0}$ & $2^{113.1}$ & $2^{54.6}$ & $2^{158.4}$ & $2^{315.3}$& $2^{38.9}$ \\
& 3 & $2^{-228.6}$ & $2^{-93.3}$ & $2^{54.6}$ & $2^{51.0}$ & $2^{287.2}$& $2^{38.9}$ \\
& 2 & $2^{-338.0}$ & $2^{-339.9}$ & $2^{54.6}$ & $2^{8.6}$ & $2^{396.6}$& $2^{38.9}$ \\
& 1 & $2^{-525.0}$ & $2^{-\infty}$ & $2^{54.6}$ & $2^{0.0}$ & $2^{583.6}$& $2^{38.9}$ \\
\midrule
\multirow{7}{*}{(251, 150)} 
& 7 & $2^{-0.0}$ & $2^{909.7}$ & $2^{73.3}$ & $2^{918.2}$ & $2^{924.2}$& $2^{52.2}$ \\
& 6 & $2^{-55.8}$ & $2^{665.8}$ & $2^{73.3}$ & $2^{675.0}$ & $2^{736.9}$& $2^{52.2}$ \\
& 5 & $2^{-121.8}$ & $2^{417.5}$ & $2^{73.3}$ & $2^{432.3}$ & $2^{560.2}$& $2^{52.2}$ \\
& 4 & $2^{-202.6}$ & $2^{160.3}$ & $2^{73.3}$ & $2^{226.7}$ & $2^{435.3}$& $2^{52.2}$ \\
& 3 & $2^{-306.8}$ & $2^{-116.9}$ & $2^{73.3}$ & $2^{81.0}$ & $2^{384.1}$& $2^{52.2}$ \\
& 2 & $2^{-453.6}$ & $2^{-447.8}$ & $2^{73.3}$ & $2^{10.9}$ & $2^{530.9}$& $2^{52.2}$ \\
& 1 & $2^{-704.6}$ & $2^{-\infty}$ & $2^{73.3}$ & $2^{0.0}$ & $2^{781.9}$& $2^{52.2}$ \\
\midrule
\bottomrule
\end{tabular}
\vspace{0.25em}
\caption{Hybrid-\ListCVP{} cost estimation for CROSS instances and a reduced instance used for experiments. The attack consists in probabilistic truncation from the initial $z=7$ to $z'$ and running sieve and enumeration algorithms. The complexities are adapted to achieve success rate of 90\%. Up to the correction, total time is equal to the product of the inverse of the truncation probability and the sum of sieve and enumeration costs.}
\label{tab:list-cvp-cost}
\end{table}

The first step in the attack is to truncate the instance from $z=7$ to $z' \le z$. According to \autoref{st:mult-truncation}, this has a probability of success $(z'/z)^n$, which means that $(z/z')^n$ attempts are needed on average. We increase this number of attempts by a factor of four to reach the success probability $1-e^{-4}\approx 98\%$. Here, we use the Poisson distribution as the limit of the binomial distribution. This defines the first factor in the final time complexity.

The second step consists in applying the \ListCVP{} reduction from \autoref{st:list-cvp-reduction} with the distance upper bound $R$ being the median of the distribution, approximated using \autoref{heur:median}. Since the median only ensures the success rate of $50\%$, we repeat the entire process four times, each time using multiplicative randomization (\autoref{st:mult-randomization}), increasing the success rate to $15/16\approx94\%$.

The third step is to solve the resulting \ListCVP{} instance by enumeration, using the complexity estimate from \autoref{st:enum}. 
The latter always includes the sieving cost, which actually dominates in instances where the target solution is expected to be unique. Due to the pruning used in \autoref{st:enum}, we repeat the process four times to reach $98\%$ success rate of the step. The total success rate is $90\%$. This step also dominates memory complexity which by \autoref{heur:cvp-cost} is $2^{0.208n+o(n)}$.

We remark that the analysis potentially omits polynomial complexity factors in sieving (which is relevant for instances degraded to \CVP{}) and in enumeration (which is relevant for \ListCVP{} cases). This is in part due to the lack of a precise analysis of the sieving complexity. We believe that this leads to an acceptable accuracy for security analysis purposes, and this simplification is in line with the common literature.

From the results, we can draw several conclusions. First, the best attacks are always obtained from the truncation to $z'=3$, which coincidentally is the largest $z'$ for which \ListCVP{} reduces to \CVP. Second, the attacks do not threaten the CROSS parameters even with potentially improved sieving and enumeration algorithms (note that, in current analysis, enumeration always incurs additional overhead in comparison to the expected number of vectors in the output). Indeed, for $z' \le 3$, the probability of successful truncation is always less than $2^{-n}$. For $z' \ge 4$, the ratio of the expected number of vectors in the considered ball $\Ball{n,R}$ to the probability of successful truncation is always larger than $2^n$.

\paragraph*{Experimental verification}
We performed experiments to verify the correctness of the reduction. The setup consists of a standard laptop with an Intel i7 CPU @ 3 GHz and 32 GB RAM. The software used is the SageMath computer algebra \cite{sagemath} and the \texttt{fpylll} library \cite{fpylll}.

For $n=35,k=21,z=4$ ($n,k$ are downscaled CROSS parameters and $z=4$ should be understood as $z'=4$ after multiplicative truncation), our complexity analysis predicts $2^{20.21}$ vectors within the mean-based center $\mu=3.75$ and the heuristic median radius $R=15.86$ (for the set $E=\pset{1,2,4,8}$).
In $10^5$ tries, we observed the image of the correct solution satisfy the bound in about 51\% of the instances, confirming \autoref{heur:median}.
Computing BKZ-$n$-reduced basis took negligible time.
Enumeration using \texttt{fpylll} without pruning took about 12 seconds (theoretical time complexity estimate $2^{32.34}$). Using exact values of the Gram-Schmidt matrix of the reduced basis (instead of the GSA), the computed complexity reduced to $2^{28.46}$ on average (this matches common observations on strong reduction quality for small dimensions).
Over several experiments, we always observed around $1.2$ million vectors within the target vector, closely matching the prediction $2^{20.21}$.
Whenever the reduction image of the initial solution satisfied the bound, it was correctly recovered among the enumerated vectors.

\section{Conclusion}
\label{sec:conclusion}

In this work, we have shown and studied reductions of the recent Restricted Syndrome Decoding problem to the more traditional Regular Syndrome Decoding problem with the particularity of light-regular errors, as well as to the lattice-based search (\SVP, \CVP) and enumeration problems (\ListSVP, \ListCVP). 
The reductions range from exact, deterministic and relatively large output instances (dimension $zn$), to heuristic, probabilistic and compact instances (dimension $n$).
We believe that our results shed new light on the \ResSD{} problem and open new research directions, such as finding more connections between the paradigms, finding new reductions, advancing analysis of the resulting instances.
For the CROSS signature, we obtained some new time-memory trade-off points, although our current reductions and attacks do not threaten its security.  

\paragraph*{Acknowledgments}
This work was funded by the Luxembourg National Research Fund (FNR), project PQseal C24/IS/18978392.

\FloatBarrier

\bibliographystyle{splncs04}
\bibliography{abbrev3,crypto,custom,codecrypto}

\appendix

\renewcommand{\theHsection}{appendix.\Alph{section}}
\renewcommand{\theHsubsection}{appendix.\Alph{section}.\arabic{subsection}}
\renewcommand{\theHsubsubsection}{appendix.\Alph{section}.\arabic{subsection}.\arabic{subsubsection}}

\renewcommand{\theHfigure}{appendix.\Alph{section}.\arabic{figure}}
\renewcommand{\theHequation}{appendix.\Alph{section}.\arabic{equation}}
\renewcommand{\theHtable}{appendix.\Alph{section}.\arabic{table}}

\section{Lattice enumeration methods}
\label{app:enum}

In the following, we summarize a standard heuristic framework for solving \ListSVP and  \ListCVP, based on enumeration of lattice points inside a given ball ($\LL \cap (\yv+\Ball{n,r})$ for a given basis of $\LL$ and $\yv$). Note that enumeration was previously developed and used for the purpose of solving SVP/CVP; the more modern and optimized sieving methods supersede enumeration. However, enumeration is still useful for solving the ``List'' problem variants.

\paragraph*{Orthogonalization}
Applying Gram-Schidmt orthogonalization to a full-rank lattice basis $\Bm = (\bv_1,\ldots,\bv_n) \in (\R^n)^n$ yields the orthogonal vectors $(\bvs_1,\ldots,\bvs_n)$ and projections $\mu_{i,j} \in \R$ of $\bv_i$ on $\bvs_j$, $1 \le i,j \le n$, satisfying $\bv_i = \sum_{j=1}^i \mu_{i,j}\bvs_j$. Define the orthogonal projections $\pi_i$, $1 \le i \le n$,
$$\pi_i:~~
\R^n \to \Span\pset{\bv_1,\ldots,\bv_{i-1}}^{\bot}:~~
\bv_k \mapsto \sum_{j=i}^n \mu_{k,j}\bvs_j.
$$
For $1 \le a \le b \le n$, we write $L_{[a,b]}$ to denote the lattice spanned by the vectors $\pi_a(\bv_a), \ldots, \pi_a(\bv_b)$.

\paragraph*{Reduced bases}
The \emph{root Hermite factor} $\delta_0$ of a vector $\xv$ in a lattice $\LL$ is defined as \[
\delta_0 \eqdef \pround{
\frac{\enorm{\xv}}{\pround{\Vol{\LL}}^{1/n}}
}^{1/n}
\]
(see \cite{EC:GamNgu08}). We are usually interested in $\delta_0$ of the first basis vector.

The lattice basis is said to be \emph{LLL-reduced} with factor $\varepsilon, 0 < \varepsilon < 1$  \cite{LLL82} if all $\mu_{i,j} \le 1/2$ and
its Gram-Schmidt orthogonalization satisfies $$
\enorm{\bvs_{i+1} + \mu_{i+1,i}\bvs_i}^2
\ge (1-\varepsilon) \enorm{\bvs_i}^2.
$$

The lattice basis is said to be \emph{BKZ-reduced} with block size $\beta \ge 2$ and factor $\varepsilon, 0 < \varepsilon < 1$ \cite{MP:SchEuc94} if it is LLL-reduced with factor $\varepsilon$ and for each $1 \le j \le n$, $\enorm{\bvs_j} = \lambda_1(L_{[j,\min(j+\beta-1,n)]})$.

The lattice basis is said to be \emph{HKZ-reduced}, if for each $1 \le j \le n$, $\enorm{\bvs_j} = \lambda_1(\pi_j(L))$. In particular, $\bv_1 = \bvs_1$ is the shortest vector in $L$.

\begin{heuristic}[Geometric Series Assumption, GSA \cite{STACS:Schnorr03}]
Let $\Bm$ be a BKZ-$\beta$- or HKZ- reduced basis and let $(\bvs_i)_i$ be its Gram-Schmidt orthogonalization with root Hermite factor $\delta_0$.
Then, for all $i$, and with $\delta = \delta_0^{-\frac{n}{n-1}}$,
\[
\enorm{\bvs_i}
\approx
\delta^{2(i-1)} \cdot\enorm{\bvs_{1}}
=
\delta^{2(i-1)} \cdot \delta_0^n \pround{\Vol{\LL}}^{1/n}
\approx
\delta^{2i-n-1}
\pround{\Vol{\LL}}^{1/n}.
\]
\end{heuristic}

\AU{Need to add more details/be more explicit, if time allows.}

We now briefly describe the generic enumeration algorithm \cite{AMS:FinPoh85,MP:SchEuc94,IWCC:HanPujSte11}.
The idea is to enumerate elements of projections $\pi_i(\LL)$ of the lattice orthogonally to basis vectors, in reverse order ($i$ going from $n$ to 1). Let $\xv^{(i)},\LL^{(i)},\yv^{(i)}$ be  projections of $\xv,\LL,\yv$ orthogonally to the span of $\bv_1,\ldots,\bv_{i-1}$. If $\xv \in \pround{\LL \cap (\yv+\Ball{n,R})}$, then $\xv^{(i)} \in \pround{\LL^{(i)} \cap (\yv^{(i)}+\Ball{n-i+1,R})}$ for all $i$. Furthermore, $\LL^{(i)} \cap (\yv^{(i)}+\Ball{n-i+1,R})$ can be obtained from  $\LL^{(i+1)} \cap (\yv^{(i+1)}+\Ball{n-i,R})$ in an efficient way. The complexity of the enumeration is bounded (up to polynomial factors) by the sum of the sizes of projected intersections $\LL^{(i)} \cap (\yv^{(i)}+\Ball{n-i+1,R})$, which by GH gives cost \[
T = \sum_{i=1}^n \frac{\Vol{\Ball{n-i+1,R}}}{\Vol{\LL^{(i)}}}.
\]
We have $\Vol{\LL^{(i)}}=\prod_{j=i}^n\enorm{\bvs_j}$ which under GSA is equal to $\delta^{i(n-i)} \pround{\Vol{\LL}}^{i/n}$.

The second component is \emph{pruning}. Observe that for most points $\xv^{(i)} \in \pround{\LL^{(i)} \cap (\yv^{(i)}+\Ball{n-i+1,R})}$  their projections will have shorter length than the vector themselves. Linear pruning consists in using the upper bound $R_d = R \cdot \sqrt{d/n}$ at the dimension $d=n+i-1$. The authors of \cite{EC:GamNguReg10} proved that the probability of a uniformly random vector on a sphere satisfying this constraint at all levels is exactly $1/n$. Therefore, randomizing the basis (see the same work) and repeating the search $cn$ times on for some constant $c$ allows us to list all but a negligible fraction of vectors. 

Finally, we consider an HKZ-reduced basis, which can be reduced to solving a polynomial time of SVP instances. Thus, the final cost includes the cost of sieving (\autoref{st:sieving}). This means
$\delta_0 = \pround{\sqrt{\frac{n}{2\pi e}}}^{1/n}$ and
$\delta = \delta_0^{-\frac{n}{n-1}}
= \pround{\sqrt{\frac{2\pi e}{n}}}^{1/(n-1)}
$.

\begin{remark}
The sieving algorithm from \autoref{st:sieving} outputs $2^{0.208n}$ short vectors. These short vectors could potentially be used to improve the enumeration complexity beyond using an HKZ-reduced basis. This is an interesting question beyond the scope of this work.
\end{remark}

\end{document}